\documentclass[journal]{IEEEtran} 
\IEEEoverridecommandlockouts
\usepackage{graphicx}                                      
\usepackage{amssymb,amsthm,bm}  
\usepackage[cmex10]{amsmath} 

\usepackage{multirow}   
\usepackage{tikz}                   
\newtheorem{theorem}{Theorem}
\newtheorem{lemma}{Lemma}
\newtheorem{corollary}[lemma]{Corollary}
\newtheorem{remark}{Remark}
\newtheorem{definition}{Definition}

\newtheorem{proposition}{Proposition}

\newtheorem{example}{Example}

\newcommand{\utag}[2]{\mathop{#2}\limits^{\text{(#1)}}}
\newcommand{\uref}[1]{(#1)}

\long\def\symbolfootnote[#1]#2{\begingroup
\def\thefootnote{\fnsymbol{footnote}}\footnote[#1]{#2}\endgroup}

\usepackage{algorithm,algpseudocode}

\usepackage{xcolor}

\newcommand{\cB}{\mathcal{B}}

\newcommand{\cD}{\mathcal{D}}

\newcommand{\cF}{\mathcal{F}}

\newcommand{\cN}{\mathcal{N}}
\newcommand{\cO}{\mathcal{O}}

\newcommand{\cQ}{\mathcal{Q}}

\newcommand{\cU}{\mathcal{U}}

\newcommand{\cZ}{\mathcal{Z}}
\usepackage{mathrsfs}

\newcommand{\sP}{\mathscr{P}}

\newcommand{\sW}{\mathscr{W}}

\newcommand{\Ie}{\textit{i.e., }}
\newcommand{\Eg}{\textit{e.g., }}
\DeclareMathOperator*{\argmax}{arg\,max}
\DeclareMathOperator*{\argmin}{arg\,min}

\usepackage{diagbox}
\usepackage{color}

\usetikzlibrary{arrows}
\usepackage{caption}
\usepackage{booktabs}
\usepackage{cite}
\usepackage{url}

\usepackage{subcaption}
\usepackage{cases}

\makeatletter
\newcommand{\mathleft}{\@fleqntrue\@mathmargin0pt}
\newcommand{\mathcenter}{\@fleqnfalse}
\makeatother

\usepackage{pgfplots}
  \pgfplotsset{compat=newest}

  \usepgfplotslibrary{patchplots}
  \usepackage{grffile}
  \definecolor{darkgreen}{RGB}{34,139,34}
  \definecolor{lightblue}{rgb}{0.00000,0.44700,0.74100}
  \definecolor{mycolor3}{rgb}{0.92900,0.69400,0.12500}%
  
  \usepackage{mathtools}

\usepackage{bbm}
\usepackage{tabu}
\usepackage{soul}

\newcommand{\pb}[1]{p\left(#1\right)}
\newcommand{\wb}[1]{w\left(#1\right)}

\algnewcommand{\algorithmicforeach}{\textbf{for each}}
\algdef{SE}[FOR]{ForEach}{EndForEach}[1]
  {\algorithmicforeach\ #1\ \algorithmicdo}
  {\algorithmicend\ \algorithmicforeach}

\usepackage{tikz}
\usetikzlibrary{er,automata,positioning,chains,fit,shapes,calc,arrows,plotmarks,arrows.meta}
\usepackage{tkz-graph}
\usepackage{caption}
\usepackage{varwidth}
\usetikzlibrary{calc}

\definecolor{darkgreen}{RGB}{34,139,34}
\definecolor{bleudefrance}{rgb}{0.19, 0.55, 0.91}

\definecolor{darkgreen2}{RGB}{242,248,241}
\definecolor{bleudefrance2}{RGB}{229, 240, 253}
\definecolor{magenta2}{RGB}{252, 241, 249}

\usepackage{transparent}

\begin{document}

\title{
ON-OFF Privacy Against Correlation Over Time
}
\author{
    Fangwei~Ye,~\IEEEmembership{Member,~IEEE},
    Carolina~Naim,~\IEEEmembership{Graduate Student Member,~IEEE} and Salim~El~Rouayheb,~\IEEEmembership{Member,~IEEE}

	\thanks{
	This paper was presented in part in ``Preserving ON-OFF Privacy for Past and Future Requests,'' IEEE
	Inf. Theory Workshop (ITW), Visby, Gotland, 2019.  This work was supported by NSF Grant CCF 1817635.
}
    \thanks{
    F.~Ye, C.~Naim and S.~El~Rouayheb are with the Department of Electrical and Computer Engineering, Rutgers, The State University of New Jersey, Piscataway, NJ 08854, USA (email: fangwei.ye@rutgers.edu; carolina.naim@rutgers.edu; salim.elrouayheb@rutgers.edu).
    }

}

\maketitle

\begin{abstract}
We consider the problem of ON-OFF privacy in which a user is interested in the latest message generated by one of $n$ sources available at a server. The user   has the choice to turn privacy ON or OFF depending on whether he wants to hide his interest at the time or not. The challenge of allowing the  privacy to be toggled  between ON and OFF  is that the user's online behavior is correlated over time. Therefore, the user cannot simply ignore the privacy requirement when privacy is OFF.

We represent the user’s correlated requests by an $n$-state Markov chain. 
Our goal is to design ON-OFF privacy schemes with optimal download rate that ensure privacy for past and future requests. 
We devise a polynomial-time algorithm to construct an ON-OFF privacy scheme. Moreover, we present an upper bound on the achievable rate.  
We show that the proposed scheme is optimal and the upper bound is tight for some special families of  Markov chains. We also give an implicit characterization of the optimal achievable rate as a linear  programming (LP).

\end{abstract}


\section{Introduction}
\label{introduction}

\subsection{Motivation}
\label{sec:motivation}
In the current data-driven world, users' information is always being collected online, and its privacy has become a significant concern. Many users wish to keep private their personal information, such as their age, sex, political views,  health disorders, etc. 
  Significant research has been devoted to study algorithms that   preserve users' privacy. Some of the proposed approaches include applying anonymization techniques\cite{Sweeney_2002}, differential privacy algorithms\cite{Dwork_2006}, and private information retrieval methods\cite{Chor_1995}.

Privacy, however, comes at a cost. Privacy-preserving algorithms  typically incur higher overheads in terms of computation, memory, and delay. These incurred costs motivate one to think of privacy as an expensive commodity and, therefore,  to allow  the user  to request it, \Ie turn privacy ON,  only when needed; otherwise, turn it OFF. The user may choose to switch between privacy being ON and OFF depending on several criteria, such as location (country, workplace vs.\ home, etc.), network connection (public or private network), devices (shared vs.\ personal machines) being used, or service quality (privacy-preserving algorithms typically induce more overheads), to name a few.

At a conceptual level, ON-OFF privacy algorithms enable privacy to be switched between ON and OFF whenever desired. One of the main challenges in designing such algorithms is correlation. For instance, a user's online behavior is personal which creates correlation over time. That is, by monitoring the user's behavior when his privacy is OFF, one may learn about the user's behavior when his privacy was ON. Therefore, the user cannot simply ignore the privacy requirement when privacy is OFF.

Take for example a user who is subscribed to two political online video channels, one is pro-right, and the other is pro-left. The user is interested in watching the latest videos posted by one of these channels. 
Correlation over time here is due to the fact that a typical user is more likely to keep watching videos from the same channel. 
One may think of a scenario where the user is more likely to watch the top item in his recommended list that depends on the previously watched videos.
Therefore, when the user switches his privacy from ON to OFF, the user cannot openly request the video he is interested in because this leaks information on what he was watching right before\,(when privacy was ON).

In this work, we abstract the previous example into the information retrieval setting, i.e., downloading messages from a server. The user may choose to turn privacy ON or OFF at each instant.
When privacy is ON, the user wants to completely hide, in an information-theoretic sense, his interest from the server. Otherwise, when privacy is OFF, the user does not worry about the privacy of his interest at that particular instant. Nevertheless, he must    be careful not to leak information about his previous or future interests that he wants to keep private. 
Our objective is to construct ON-OFF privacy schemes that:
\begin{enumerate}
\item Deliver to the user his request while ensuring perfect information-theoretic privacy against the server.   That is, the observations of the server  must be statistically independent of the user's interests when Privacy is ON. 
\item Maximize the download rate or equivalently minimize the amount of downloaded information.
\end{enumerate}

\subsection{Related Work}
The study of information-theoretic measures for privacy has received significant interest in the literature (see for e.g.
  \cite{Sankar_privacy_utility,Gertner_2000,Bezzi_2010,Nekouei_2019,Flavio_2012}).
  The closest problem to the ON-OFF privacy problem studied in this paper is the private information retrieval\,(PIR) from a single server\cite{Chor_1995}, which can be viewed as a special case (when privacy is always ON) of the ON-OFF privacy problem. 
In this case, it is known that to achieve information-theoretic privacy, the user must download all the messages, except in the case when the user has some side information \cite{Kadhe_2017,Li_side}. Recently, there  has been significant progress on PIR with multiple servers with a focus on download rate and coded data (\Eg \cite{Shah_2014,  Sun_2017, Tajeddine_2016, Freij-Hollanti_2017, Banawan_2018} and references therein).

A related problem that considers privacy with correlation, namely location privacy, was studied in \cite{Shokri_2011,Chatzikokolakis_2014,Xiao_2015,Shokri_2017,trace_privacy,Gunduz_2021,Gedik_anonymity,k-anonym_2010,Tandon}. 
The privacy notions studied therein include
$k$-anonymity \cite{Gedik_anonymity,k-anonym_2010}, (extended) differential privacy \cite{trace_privacy,Chatzikokolakis_2014,Xiao_2015}, and distortion privacy  \cite{Shokri_2011,Shokri_2017}, which all differ from the information-theoretic privacy measure studied in this paper. The works of  \cite{Gunduz_2021,Tandon} recently studied the information-theoretic privacy measure in location-privacy protection mechanisms, and their privacy metric was defined by the mutual information between the released data and the true traces. In this paper's language, it can be viewed as the case when privacy is always ON. 
However, in this paper, we want to prevent the adversary from inferring a selective part of the requests specified by an ON or OFF privacy status, and the simple time-sharing (switching between a private and a non-private scheme according to the privacy status) approach is not permissible due to the correlation.

The ON-OFF privacy problem was studied by the authors first in \cite{Naim_2019}. The focus was on preserving the privacy of past requests for which privacy was ON. This paper is based on  the setting studied later in \cite{Ye_2019}   which   requires   privacy of both past and future requests. The work in  \cite{Ye_2019} studied
the special case when $n=2$ sources, and an optimal scheme and a tight upper bound on the rate  were presented therein. The concept of ON-OFF privacy was also applied to preserve privacy of sensitive genotypes in genomics in  \cite{Genotype_2020}.

\subsection{Contributions}

To study how correlation affects privacy, we focus in this paper on the simplest non-trivial correlation model given by a Markov chain. That is, we assume that  the user's requests to the server  are correlated in time according to a Markov chain. We also assume that  the user knows his future requests into a window of size $\omega$\footnote{This can happen in applications where  the user places his requests in a queue of size $\omega$. For example, the user may know what he will be watching next since it is the next item in a playlist or the top recommendation in a recommended list.}.

 Under this model, our main result is summarized in Theorem~\ref{theorem:main} which: (i) gives a general upper bound on the download rate;  and (ii) gives an achievable rate obtained by an ON-OFF privacy scheme having  polynomial time complexity in $n$. 
 
 We show that our proposed scheme is optimal, \Ie the upper bound is tight, for a   family of  Markov chains for $n>2$.
 
  For $n=2$ sources, this scheme is equivalent to the one in \cite{Ye_2019} and therefore is always optimal.  Therefore, the results in this paper can be viewed as a   generalization of  the earlier results in  \cite{Ye_2019} on $n=2$ sources   to  any $n\geq 2$ sources.

We also give an implicit characterization of 
the optimal achievable rate, which relies on solving a linear program (LP) with an exponential number (in $n$) of variables and constraints. Thus, it is intractable to tackle it using standard LP solvers (\Eg \cite{LP-book}).    
From that perspective,   our results can be viewed  as leveraging the special structure of the problem to provide  an efficiently computable upper bound and a polynomial time scheme.

\subsection{Organization}

The rest of the paper is organized as follows. In Section~\ref{section:formulation}, we describe the formulation of the ON-OFF privacy problem. 
We present our main result, Theorem~\ref{theorem:main}, in Section~\ref{section:result}, and 
its corollaries in Section~\ref{sec:Special_Markov}. 
In Sections \ref{section:scheme} and \ref{section:algorithm_achievable}, we propose  an efficient ON-OFF privacy scheme that gives the achievable bound in Theorem~\ref{theorem:main}. In Section~\ref{section:converse}, we derive the  upper bound, in Theorem~\ref{theorem:main},  on the achievable rate. Finally, we present an implicit characterization of the optimal achievable rate 
in Section~\ref{section:discussion}. We conclude   in Section~\ref{section:conclusion}.



\section{Problem Formulation}
\label{section:formulation}
\subsection{System Model}
A single server stores $n$ information sources  $\{\sW_i: i\in \cN\}$, where $\cN:=\{1,2,\ldots,n\}$. 
The system is time-varying, and the time index $t$ is assumed to be discrete, \Ie $t \in \mathbb{N}$, throughout this paper. At each time $t$, each source $\sW_i$ generates a new message $W_{i,t}$ of length $L$, which is independent of previously generated messages $\left\{W_{i,j}: j=0,\ldots,t-1\right\}$. Without loss of generality, we assume that $W_{i,t}$ for $i \in \cN$ and $t \in \mathbb{N}$ are independently and identically drawn from the uniform distribution over $\{0,1\}^{L}$. 

At time $t$, the user is interested in retrieving the latest message generated by a desired source, \Ie one of the messages from $\{W_{i,t}: i \in \cN \}$. In particular, let $X_t$ be the source of interest at time $t$, which takes values in $\mathcal{N}$. In the sequel, we will call $X_t$ the \emph{user's request} at time $t$. For notational simplicity, we drop $t$ from $W_{i,t}$ when the time index $t$ is clear from context, \Ie $W_{i,t}$ will be denoted by $W_i$. To retrieve the desired message, the user is allowed to construct a query $Q_t$ and send this query to the server. Upon receiving the query, the server responds to the user by producing an answer $A_t$. After receiving the answer, the user should be able to recover the message $W_{X_t}$ that he is interested in. 

Meanwhile, the user may wish to hide the identity of his source of interest at time $t$. Specifically, the user may choose the \emph{privacy status} $F_t$ to be ON or OFF. When $F_t$ is ON, the user wishes to keep $X_t$ private and when $F_t$ is OFF, the user is not concerned with hiding $X_t$. We assume that the privacy status $\{F_t:t \in \mathbb{N}\}$ is independent of the user's requests, and the user's privacy status $\{F_i: i \leq t\}$ is known and recorded by both the server and the user at time $t$. 

We assume in our model  that the privacy status is independent of the user's requests because as mentioned in Section~\ref{sec:motivation}, the user may choose privacy to be ON or OFF depending on many factors such as location, network connection, devices or service quality etc, and in general these factors are independent of the user's requests.

In this paper, we are particularly interested in the case where the requests $X_t$ form a Markov chain, \Ie $\{X_t: t \in \mathbb{N}\}$ is generated by a (discrete) Markov source. The transition matrix $P$ of the Markov chain is known by both the server and the user, and the transition probability from state $i$ to state $j$ is denoted by $P_{i,j}$.

Moreover, we assume that the user knows his future requests in a window of positive size $\omega$\footnote{If the window size $\omega=0$, \Ie no future requests are known, 
we have to relax the the stringent privacy requirement \eqref{eq:privacy} defined in this work to a  weaker  sense  where  only  past  requests  are  protected.   This  falls  into  a  different model studied in \cite{Ye_2020}.}, 
which means that at time $t$, the user knows the future requests $\{X_{t+1},\ldots,X_{t+\omega}\}$ in addition to the current and all past requests $\{X_0,\ldots,X_t\}$. 
This models several scenarios where user's requests are in a queue. One can think of the situation where the user places
his requests in a playlist when watching videos.

The system mainly consists of two encoding functions, which we describe below. Let $[t]$ denote $\{0,1,\ldots,t\}$ and $X_{[t]}$ denote  $\{X_0,X_1,\dots,X_t\}$ for $t\in \mathbb{N}$ in the sequel.

\noindent \textbf{Query encoding function:}
The query $Q_t$, at time $t$, is generated by a query encoding function $\phi_t$. Given the assumptions that the messages $W_{i,t}$ (as well as the answers $A_t$) are independent over time and the privacy status $F_{[t]}$ are known by both the user and the server, we suppose that $\phi_t$ is a probabilistic function of the user's known requests\footnote{
One may also take all previous queries $Q_{[t-1]}$ as variables of the function. However, since $Q_{[t-1]}$ is also a probabilistic function of $X_{[t+\omega]}$, the variables of the function can be written as in \eqref{eq:query-encoding-function}.} 
$X_{[t+\omega]}$ for some $\omega\in\mathbb{N}^{+}$, \Ie 
\begin{equation}
\label{eq:query-encoding-function}
	Q_t = \phi_t \left( X_{[t+\omega]},\mathsf{K} \right),
\end{equation}
where $\mathsf{K}$ is the random key to generate a probabilistic query.

\noindent \textbf{Answer encoding function:}	
Accordingly, the answer $A_t$ from the server is given by the answer encoding function $\rho_t$, which 
is assumed to be a deterministic function of the query $Q_t$ and the latest messages, \Ie 
\begin{equation}
 	A_t = \rho_t \left(Q_t,W_{1},\ldots, W_{n}  \right).
\end{equation} 
In particular, the length of answer $A_t$ is assumed to be a function of the query $Q_t$, and we denote this length by $\ell(Q_t)$. Then, the average length of the answer $A_t$ is given by 
\begin{equation}
  \ell_{t} = \mathbb{E}_{Q_{t}} [\ell\left(Q_t\right)],
\end{equation}
where $\mathbb{E[\cdot]}$ is the expectation operator.

After receiving the answer $A_t$, the user should be able to recover the desired message from the answer with zero-error probability. This is referred to as the \emph{decodability} condition.

\subsection{Adversary Model}

The adversary is the untrusted server and is assumed to have full statistical knowledge of user's requests and the querying mechanism, that is, the Markov chain's transition probabilities modeling the user's requests and the querying mechanism that generates the queries for information retrieval, respectively. 

We assume that the server has no memory constraint, so the server can use all the queries it received up to time $t$, represented by $Q_{[t]}$, and the statistical knowledge of user's requests and the querying mechanism, to infer user's private requests, \Ie all previous requests of which privacy was ON and all future requests. We also assume that the adversary has unbounded computational power and can launch any attack to infer any of the user's private requests.

Privacy is quantified by the mutual information between the user's private requests and the queries released to the server. It is worth noting that the information-theoretic privacy measure is preferable in this paper, since it is independent of specific attacking strategies. 
We consider the most stringent privacy constraint, namely information-theoretic perfect privacy, which requires that absolutely zero information, measured by the mutual information, about the user's private requests is leaked to the server. Formally, it can be written as 
\begin{equation}
	\label{eq:privacy}
		I \left(X_{\cB_t};Q_{[t]}\right) = 0, ~ \forall t\in \mathbb{N},
\end{equation}
where $\cB_t := \{i:  i \leq t, F_i=\text{ON}\} \cup \{i: i \geq t+1\}$ denotes the user's private requests, \Ie all previous requests of which privacy was ON and all future requests, and $I(\cdot)$ denotes the mutual information. We refer to \eqref{eq:privacy} as the \emph{privacy} condition.


\begin{remark}
The privacy requirement in \eqref{eq:privacy} implies that at time $t$, only the previous privacy status $\{F_i : i \leq t\}$ is known, and the user may not know whether he will choose privacy to be ON or OFF in the future. For this reason, we have adopted a worst-case formulation in the privacy constraint by assuming that privacy is always ON in the future. In other words, at time $t$, all previous requests when privacy was ON, as well as all future requests need to be protected. This is characterized by the set $\cB_t:= \{i:  i \leq t, F_i=\text{ON}\} \cup \{i: i \geq t+1\}$ in \eqref{eq:privacy}.
\end{remark}

For large messages, the upload cost is negligible relative to the download cost, so in this paper, we are interested in minimizing the download cost of the answer at each time, \Ie the average length $\ell_t$ at time $t$. By convention, we measure the efficiency by the download rate $R_t= L/\ell_t$, and define the achievable rate region as follows.
\begin{definition}[Achievable Rate]
The rate tuple $\left(R_t:t \in \mathbb{N}\right)$ is achievable if there exists a scheme with average download cost $\ell_t$ 
such that $R_t \leq L /\ell_t$.
\end{definition}

In the rest of this paper, we will study the achievable region of $\left(R_t:t \in \mathbb{N}\right)$. In particular, the focus of this paper is the characterization of $R_t$ for each $t \in \mathbb{N}$.

\section{Main Result}
\label{section:result}
Before stating the main result, we introduce some necessary notation. Let $\tau(t)$ be the last time privacy was ON, \Ie 
\begin{equation}
\label{eq:def-tau}
	\tau(t):= \max\{i: i \leq t, F_i = \text{ON}\}.
\end{equation}
Without loss of generality, we assume that $F_0 = \text{ON}$, so $\tau(t)$ is always well-defined.
Also, when the time index $t$ is clear from context, we drop $t$ from the notation and write $\tau(t)$ as $\tau$ for simplicity. For our analysis, it is convenient to define 
\begin{equation}
\label{eq:def-u}
	U_t := \left(X_{\tau},X_{t+1}\right),
\end{equation}
which represents the last request when privacy was ON and the next request of the user at time $t$, so the alphabet size of $U_t$ is $\cN^2$. 

Stating our main results calls for the following notation, which we summarize in Figure \ref{fig:matrix_variables}. The inherent value of this notation will be apparent when we give the proofs of our main results in later sections. 
For any given $x \in \cN$, we can order the likelihood probabilities $\pb{X_t=x|U_t=u_{x,i}}$ such that
\begin{multline}
\label{eq:def-order}
 \pb{X_t=x|U_t = u_{x,1}} \leq \pb{X_t=x|U_t = u_{x,2}} \\
\leq \cdots \leq \pb{X_t=x|U_t = u_{x,m}},
\end{multline}
where $m=n^2$ and $u_{x,i}$ for $i=1,\ldots,m$ are distinct elements in $\cN^2$. Note that probabilities $\pb{X_t=x|U_t=u}$ for $x \in \cN$ and $u \in \cN^2$ can be determined by the given Markov chain. 
These ordered probabilities can be stored in the columns of a matrix, as shown in Figure~\ref{fig:matrix_variables}.
Then, for $x \in \cN$ and $i=1,\ldots,m$, let $\lambda_i(t)$ be the summation of row $i$ of this matrix, more formally,
\begin{equation}
    \label{eq:def-lambda}
	\lambda_{i}\left(t\right) =  
	 \sum_{x \in \cN} \pb{X_t=x|U_t=u_{x,i}}.
\end{equation}

\begin{figure}
 \centering
 \scalebox{0.88}{
     \begin{tikzpicture}
  
\matrix[ampersand replacement=\&] {
        \node (tab) {
        \renewcommand{\arraystretch}{1.5}
            \begin{tabular}{lccccc}
           &$1$&$2$&$\cdots$&$n$&\\
 $u_{x,1}$&$p(1|u_{1,1})$ & $p(2|u_{2,1})$  & $\cdots$ & $p(n|u_{n,1})$ \\
$u_{x,2}$&$p(1|u_{1,2})$ & $p(2|u_{2,2})$  & $\cdots$ & $p(n|u_{n,2})$ \\ 
\hspace{10pt}$\vdots$ &$\vdots$&$\ddots $ &  &$\vdots$ \\ 
$u_{x,m}$&$p(1|u_{1,m})$ & $p(2|u_{2,m})$  & $\cdots$ & $p(n|u_{n,m})$
\end{tabular}
        };
        \& 
        \node {}; \\      
        \\ 
};

\draw[- , thick, gray] (-2.9,-1.7) -- (-2.9,1);
\draw[- , thick, gray] (-2.9,-1.7) -- (-2.7,-1.7);
\draw[- , thick, gray] (-2.9,1) -- (-2.7,1);

\draw[- , thick, gray] (3.4,-1.7) -- (3.4,1);
\draw[- , thick, gray] (3.4,-1.7) -- (3.2,-1.7);
\draw[- , thick, gray] (3.4,1) -- (3.2,1);
\end{tikzpicture}

    }
    \caption{For a given $x$, sort the probabilities $\pb{X_t=x|U_t=u}$ for $u\in\cN^2$ in an ascending order, and store the values in column $x$ where $m=n^2$. $\lambda_i$ is the sum of the $i^{th}$ row.}
    \label{fig:matrix_variables}
\end{figure}

Also, for $i=1,\ldots,n$, let
\begin{equation}
\label{eq:def-theta}
	\theta_{i}(t) = 
\begin{cases}
	\lambda_i(t)- \lambda_{i-1}(t), & i < n, \\
	1 - \lambda_{i-1}(t),			& i = n, 
\end{cases}
\end{equation}
where $\lambda_{0}\left(t\right)$ is assumed to be $0$. 
For notational simplicity, let 
\begin{equation}
\label{eq:notation_inner}
	\frac{1}{R_t^{I}}:=  \sum_{i=1}^{n} i \, \theta_i(t),
\end{equation}
and
\begin{equation}
\label{eq:notation_outer}
	\frac{1}{R_t^{O}}:= \lambda_m(t) = \sum_{x\in\cN} \max_{u\in\cN^2} p(X_t=x|U_t=u), 
\end{equation}
where $\lambda_m(t)$ is defined in \eqref{eq:def-lambda}. We may drop the time index $t$ when it is clear from context, that is, we will write $\lambda_i\left(t\right)$ and $\theta_i\left(t\right)$ as $\lambda_i$ and $\theta_i$, respectively. 
With this notation, we are ready to state the main theorem.
\begin{theorem}
\label{theorem:main}
Suppose that $\left\{X_t: t \in \mathbb{N}\right\}$ is a Markov process with the transition matrix $P$. The rate tuple $\left(R_t:t \in \mathbb{N}\right)$ is achievable if 
\begin{equation}
\label{eq:thm-inner}
	\frac{1}{R_t} \geq \frac{1}{R_t^{I}},
\end{equation}
where $1/R_t^{I}$ is defined in \eqref{eq:notation_inner}. On the other hand, any achievable rate tuple  $\left(R_t:t \in \mathbb{N}\right)$ must satisfy 
\begin{equation}
\label{eq:thm-outer}
	\frac{1}{R_t} \geq \frac{1}{R_t^{O}},
\end{equation}
where $1/R_t^{O}$ is defined in \eqref{eq:notation_outer}.
\end{theorem}

\begin{remark}   [Single server PIR]
As mentioned earlier, the single server private information retrieval problem can be viewed as a special case of this setting where privacy is always ON. As a sanity check, if $F_t = \text{ON}$ for all $t\in\mathbb{N}$, we have $U_t=\left(X_t,X_{t+1}\right)$ by the definition \eqref{eq:def-u}, and then we can easily see that $ \max_{u_t} \pb{x_t|u_{t}} = 1$, for all  $x_t \in \cN$. Thus, we know from Theorem~\ref{theorem:main} that $R_t$ is achievable only if
\begin{equation*}
	R_t \leq R_t^{O} = \frac{1}{\lambda_m(t)} = \frac{1}{n},
\end{equation*}
which implies that it is necessary to download all messages
when the privacy is ON. This is consistent with the well- known result in the literature on PIR\cite{Chor_1995}.
\end{remark}

The rest of the paper is dedicated to proving Theorem~\ref{theorem:main}. 
In particular, we propose a polynomial-time querying scheme that achieves $R_t^{I}$ in sections \ref{section:scheme} and \ref{section:algorithm_achievable}. As discussed in the previous remark, the user has to query for all the messages when privacy is ON, so our focus will be on the instances when privacy is OFF. Roughly speaking, in our proposed probabilistic querying scheme, the user asks for a subset of the messages containing the message in which he is interested. The user generates his query based on his knowledge of his previous requests when privacy was ON, his current request, and his next request.   
Moreover, the proof of the upper bound $R_t^{O}$ will be presented in Section~\ref{section:converse}.

\section{Optimality for Special Families of Markov Chains}
\label{sec:Special_Markov}

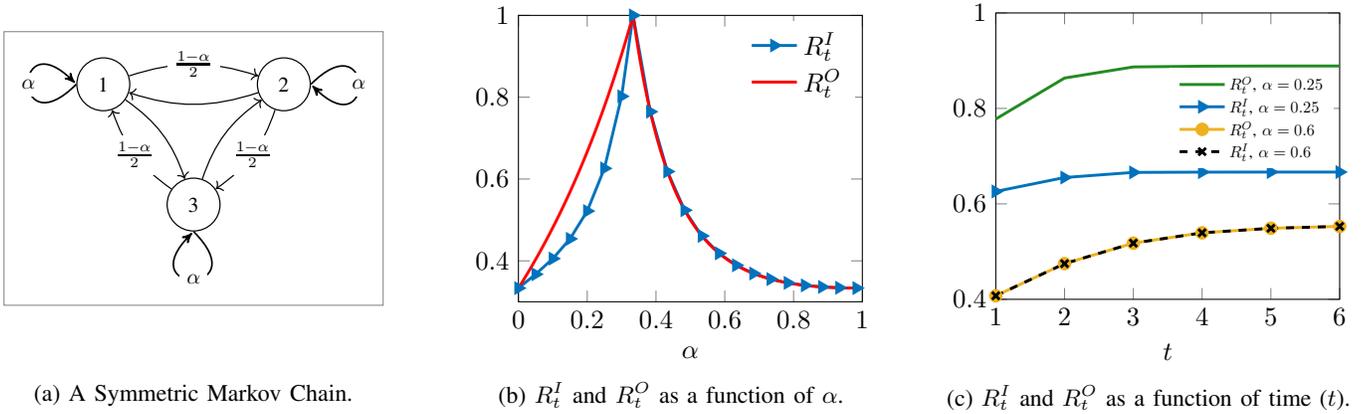
\begin{figure*}[t]
     \centering
     \begin{subfigure}[l]{0.3\textwidth}
         \centering
          \scalebox{0.8}{
                  \begin{tikzpicture}[thick,scale=1, every node/.style={scale=1}]
 \tikzstyle{every node}=[font=\normalsize]
        \tikzset{node style/.style={state, 
                      inner sep=0.5pt,
                                    minimum size = 25pt,
                                    line width=0.2mm,
                                    fill=white},
                    LabelStyle/.style = { minimum width = 1em, fill = white!10,
                                            text = black},
                   EdgeStyle/.append style = {->, bend left=18, line width=0.2mm} }
        
        \node at (6.5,9) [rectangle,minimum width=1cm, minimum height=4.8cm] {};
        
        \node at (6.5,8.6) [line width=0.1pt, rectangle,draw,minimum width=6.3cm, minimum height=4.55cm,color=gray] {};

        \node[node style] at (5,10)     (1)     {1};
        \node[node style] at (8, 10)     (2)     {2};
    \node[node style] at (6.5, 8)     (3)     {3};

            \Edge[label=$\frac{1-\alpha}{2}$](1)(2)
        \Edge (2)(1)
        \Edge (1)(3)
        \Edge[label=$\frac{1-\alpha}{2}$](3)(1)
        \Edge (3)(2)
        \Edge[label=$\frac{1-\alpha}{2}$](2)(3)
            
            \Loop[dist = 1.5cm, dir = NO, label = $\alpha$](1.west)
            \Loop[dist = 1.5cm, dir = SO, label = $\alpha$](2.east)           
            \Loop[dist = 1.5cm, dir = WE, label = $\alpha$](3.south)     
            
            \node at (6.5,6.8) [rectangle,minimum width=1cm, minimum height=3cm] {};
         
    \end{tikzpicture}
         }
         \caption{\centering A Symmetric Markov Chain.}
         \label{fig:example_3states_chain}
     \end{subfigure}
     \hfill
     \begin{subfigure}[l]{0.3\textwidth}
         \centering
         
%
%
\begin{tikzpicture}

\begin{axis}[%
width=1.8in,
height=1.5in,
scale only axis,
xmin=0,
xmax=1,
ymin=0.3,
ymax=1,
xlabel=$\alpha$,
axis background/.style={fill=white},
legend style={legend cell align=left, align=left, draw=none}
]
\addplot [color=lightblue, line width=1.1pt, mark=triangle*, mark options={solid, rotate=270, fill=lightblue, lightblue}]
  table[row sep=crcr]{%
0 0.333333333333333\\
0.05  0.367295763494787\\
0.1 0.405486659150695\\
0.15  0.453841056226742\\
0.2 0.521739130434783\\
0.25  0.626016260162602\\
0.3 0.802142407057341\\
0.333333333333333 1\\
0.383333333333333 0.764650283553875\\
0.433333333333333 0.618343195266272\\
0.483333333333333 0.523781212841855\\
0.533333333333333 0.4609375\\
0.583333333333333 0.418367346938776\\
0.633333333333333 0.389196675900277\\
0.683333333333333 0.369125520523498\\
0.733333333333333 0.355371900826446\\
0.783333333333333 0.346084200995926\\
0.833333333333333 0.34\\
0.883333333333333 0.33624065503738\\
0.933333333333333 0.334183673469388\\
0.983333333333333 0.333381212295317\\
};
\addlegendentry{$ R_t^I$}

\addplot [color=red,line width=1.1pt]
  table[row sep=crcr]{%
0 0.333333333333333\\
0.01  0.346801346801347\\
0.02  0.360544217687075\\
0.03  0.374570446735395\\
0.04  0.388888888888889\\
0.05  0.403508771929825\\
0.06  0.418439716312057\\
0.07  0.433691756272401\\
0.08  0.449275362318841\\
0.09  0.465201465201465\\
0.1 0.481481481481482\\
0.11  0.49812734082397\\
0.12  0.515151515151515\\
0.13  0.532567049808429\\
0.14  0.550387596899225\\
0.15  0.568627450980392\\
0.16  0.587301587301587\\
0.17  0.606425702811245\\
0.18  0.626016260162602\\
0.19  0.646090534979424\\
0.2 0.666666666666667\\
0.21  0.687763713080169\\
0.22  0.70940170940171\\
0.23  0.731601731601732\\
0.24  0.754385964912281\\
0.25  0.777777777777778\\
0.26  0.801801801801802\\
0.27  0.82648401826484\\
0.28  0.851851851851852\\
0.29  0.87793427230047\\
0.3 0.904761904761905\\
0.31  0.932367149758454\\
0.32  0.96078431372549\\
0.33  0.990049751243781\\
0.333333333333333 1\\
0.343333333333333 0.943020077292865\\
0.353333333333333 0.891598433606266\\
0.363333333333333 0.84508879723929\\
0.373333333333333 0.802933673469388\\
0.383333333333333 0.764650283553875\\
0.393333333333333 0.7298190175237\\
0.403333333333333 0.69807390205587\\
0.413333333333333 0.669094693028096\\
0.423333333333333 0.64260028520057\\
0.433333333333333 0.618343195266272\\
0.443333333333333 0.596104923964045\\
0.453333333333333 0.575692041522491\\
0.463333333333333 0.556932870969411\\
0.473333333333333 0.539674667724658\\
0.483333333333333 0.523781212841855\\
0.493333333333333 0.509130752373996\\
0.503333333333333 0.495614227446165\\
0.513333333333333 0.483133749367516\\
0.523333333333333 0.47160128199927\\
0.533333333333333 0.4609375\\
0.543333333333333 0.451070796793255\\
0.553333333333333 0.441936420380316\\
0.563333333333333 0.433475718637303\\
0.573333333333333 0.425635478637101\\
0.583333333333333 0.418367346938776\\
0.593333333333333 0.411627319782856\\
0.603333333333333 0.405375293794451\\
0.613333333333333 0.399574669187146\\
0.623333333333333 0.394191998627356\\
0.633333333333333 0.389196675900277\\
0.643333333333333 0.38456065934656\\
0.653333333333333 0.380258225739275\\
0.663333333333333 0.376265750864877\\
0.673333333333333 0.3725615135771\\
0.683333333333333 0.369125520523498\\
0.693333333333333 0.365939349112426\\
0.703333333333333 0.362986006603625\\
0.713333333333333 0.360249803476286\\
0.723333333333333 0.357716239461445\\
0.733333333333333 0.355371900826446\\
0.743333333333333 0.353204367672787\\
0.753333333333333 0.351202130158979\\
0.763333333333333 0.349354512690452\\
0.773333333333333 0.347651605231867\\
0.783333333333333 0.346084200995926\\
0.793333333333333 0.344643739848881\\
0.803333333333333 0.343322256848195\\
0.813333333333333 0.342112335393711\\
0.823333333333333 0.341007064531463\\
0.833333333333333 0.34\\
0.843333333333333 0.339085128653783\\
0.853333333333333 0.3382568359375\\
0.863333333333333 0.337509876119915\\
0.873333333333333 0.336839345026514\\
0.883333333333333 0.33624065503738\\
0.893333333333333 0.335709512140789\\
0.903333333333333 0.335241894854373\\
0.913333333333333 0.334834034844691\\
0.923333333333333 0.334482399092911\\
0.933333333333333 0.334183673469388\\
0.943333333333333 0.333934747593302\\
0.953333333333333 0.333732700865568\\
0.963333333333333 0.33357478957388\\
0.973333333333333 0.33345843497842\\
0.983333333333333 0.333381212295317\\
0.993333333333333 0.33334084050268\\
};
\addlegendentry{$R_t^O$}

\end{axis}
\end{tikzpicture}%
         
         \caption{\centering $R_t^I$ and $R_t^O$ as a function of $\alpha$.}
         \label{fig:example_3states_p}
     \end{subfigure}
     \hfill
      \begin{subfigure}[l]{0.3\textwidth}
         \centering
%
%
\begin{tikzpicture}

\begin{axis}[%
width=1.8in,
height=1.5in,
scale only axis,
xmin=1,
xmax=6,
ymin=0.4,
ymax=1,
xlabel=$t$,
ytick={0.4, 0.6, 0.8, 1},
xtick={1, 2, 3,4,5, 6},
axis background/.style={fill=white},
legend style={at={(0.985,0.81)},legend cell align=left, align=left, draw=none ,nodes={scale=0.61, transform shape}}
]

\addplot [color=darkgreen, line width=1.1pt]
  table[row sep=crcr]{%
1 0.777777777777778\\
2 0.863636363636364\\
3 0.886939571150097\\
4 0.888482186432406\\
5 0.888858372242058\\
6 0.888882531108619\\
};
\addlegendentry{$R_t^O$, $\alpha=0.25$}

\addplot [color=lightblue, line width=1.1pt,mark=triangle*, mark options={solid, rotate=270, fill=lightblue, lightblue}]
  table[row sep=crcr]{%
1 0.626016260162602\\
2 0.655172413793103\\
3 0.665935940376163\\
4 0.666483617060223\\
5 0.666655222989587\\
6 0.66666380565736\\
};
\addlegendentry{ $R_t^I$, $\alpha=0.25$}

\addplot [color=mycolor3, line width=1.1pt,mark=*, mark options={solid, rotate=270, fill=mycolor3, mycolor3}]
  table[row sep=crcr]{%
1 0.407407407407407\\
2 0.474747474747475\\
3 0.517730496453901\\
4 0.539320142059868\\
5 0.548865893174454\\
6 0.552847076747286\\
};
\addlegendentry{$R_t^O$, $\alpha=0.6$}

\addplot [color=black, line width=1.1pt,mark=x, dashed, mark options={solid, rotate=270, fill=black, black}]
  table[row sep=crcr]{%
1 0.407407407407407\\
2 0.474747474747475\\
3 0.517730496453901\\
4 0.539320142059868\\
5 0.548865893174454\\
6 0.552847076747286\\
};
\addlegendentry{$R_t^I$, $\alpha=0.6$}

\end{axis}

\end{tikzpicture}%

         \caption{$R_t^I$ and $R_t^O$ as a function of time ($t$).}
         \label{fig:example_3states_t}
     \end{subfigure}
     
        \caption{In Figure~\ref{fig:example_3states_chain}, we graphically represent the 3-state symmetric Markov chain used in Example~\ref{ex:symmetric}, where $0 \leq \alpha\leq 1$. In Figure~\ref{fig:example_3states_p}, we plot the achievable rate $R_t^I$(c.f.\eqref{eq:thm-inner}) and the upper bound $R_t^O$(c.f.\eqref{eq:thm-outer}),
         as a function of $\alpha$, when $\tau=0$ and $t=1$.
         In Figure \ref{fig:example_3states_t}, we plot
         $R_t^I$ and $R_t^O$
         as a function of time $t$ for both
         $\alpha=0.25$ and $\alpha=0.6$.
         }
        \label{fig:example_3states}
\end{figure*}
 
Before we proceed to prove Theorem~\ref{theorem:main}, we give two corollaries that characterize two special classes of Markov chains for which the bounds in Theorem~\ref{theorem:main} are tight, \Ie $R_t^{I} = R_t^{O}$, which means that our proposed scheme is optimal for these special cases.

\begin{corollary}[Optimality for $n=2$]
\label{corollary:two}
For the case $n=2$, the two bounds  \eqref{eq:thm-inner} and \eqref{eq:thm-outer} match, \Ie 
\begin{equation}
 	\frac{1}{R_t^{I}} = \frac{1}{R_t^{O}} = \lambda_m(t).
\end{equation} 
In other words, the rate tuple $\left(R_t:t \in \mathbb{N}\right)$ is achievable if and only if
\begin{equation}
	 R_t \leq \frac{1}{\lambda_m(t)}.
\end{equation} 
\end{corollary}

\begin{definition}[Symmetric Markov Chain]
\label{def:sym_markov_chain}
A Markov chain is symmetric if its transition matrix $P$ is given by
\begin{equation}
    P_{i,j} =
    \begin{cases}
        \alpha, & \quad \text{if} \quad i=j,\\
        \frac{1-\alpha}{n-1}, & \quad \text{if} \quad i\neq j,
    \end{cases}
\end{equation}
where $0 \leq \alpha \leq 1$ and $P_{i,j}$ denotes 
the transition probability from state $i$ to state $j$. 
\end{definition}

\begin{corollary}[Optimality for Symmetric Markov Chain]
\label{cor:sym_markov}
For the symmetric Markov chain such that $\frac{1}{n} \leq \alpha \leq 1$, 
the two bounds \eqref{eq:thm-inner} and \eqref{eq:thm-outer} match. In particular,
\begin{equation}
\label{eq:cor_sym_markov}
\frac{1}{R_t^{I}} = \frac{1}{R_t^{O}}=\alpha n\frac{ (n-1)^{t-\tau}+ (n-1)(n\alpha-1)^\tau }{(n-1)^{t-\tau}+(n\alpha-1)^{t-\tau+1}}.
\end{equation}
In other words, the rate tuple $\left(R_t:t \in \mathbb{N}\right)$ is achievable if and only if
\begin{equation}
	\frac{1}{R_t} \geq \alpha n\frac{ (n-1)^{t-\tau}+ (n-1)(n\alpha-1)^\tau }{(n-1)^{t-\tau}+(n\alpha-1)^{t-\tau+1}}.
\end{equation}
\end{corollary}
 
 The proofs of Corollaries~\ref{corollary:two} and \ref{cor:sym_markov}  can be found in Appendix~\ref{Appendix:proof_n=2} and Appendix~\ref{Appendix:proof_sym}, respectively.

\begin{example}
\label{ex:symmetric}
We study a special case described in Corollary~\ref{cor:sym_markov}. Suppose that we are 
given $\tau=0$, and a 3-state Markov chain, as represented in Figure~\ref{fig:example_3states_chain}, where $0 \leq \alpha \leq 1$.

In this case, we have two regimes, one for $\alpha<\frac{1}{3}$ and the other for $\alpha\geq\frac{1}{3}$. This is because the ordering of probabilities (c.f.\eqref{eq:def-order}) changes at $\alpha = \frac{1}{3}$.  

For $\alpha \geq \frac{1}{3}$, the bounds 
\eqref{eq:thm-inner} and \eqref{eq:thm-outer} match, \Eg 
for $t=1$,
\begin{equation}
\label{eq:exp-1}
    \frac{1}{R_1^O}=\frac{1}{R_1^I}=\frac{6\alpha^2}{3\alpha^2-2\alpha+1}.
\end{equation}
However, for $\alpha < \frac{1}{3}$ and $t=1$, we have
\begin{equation}
\label{eq:exp-2}
    \frac{1}{R_1^O}=\frac{3-3\alpha}{3\alpha+1}\leq\frac{1}{R_1^I}=\frac{2}{3 \alpha+1} - \frac{4\alpha-2}{3 \alpha^2-2\alpha+1}-1.
\end{equation}
We illustrate \eqref{eq:exp-1} and \eqref{eq:exp-2} in 
Figure~\ref{fig:example_3states_p}.

In Figure~\ref{fig:example_3states_t}, we analyze the rate over time for $\alpha=0.25$, and $\alpha=0.6$. It is notable that as $t$ grows, the correlation between $X_t$ (the current request) and $X_{\tau}$ (the request when privacy was last ON)  decreases, which leads to an increase in the download rate $R_t$.

\end{example}

\section{Achievability: Linear Programming Formulation}
\label{section:scheme}

Towards finding an ON-OFF privacy scheme, we consider uncoded queries for retrieving messages, \Ie  
the query $Q_t$ at time $t$ takes values in the power set of $\cN$, denoted by $\sP\left(\cN\right)$. In other words, the user will query for a subset of the messages $W_\cN$ at each time. Later in this section, we will see that designing an uncoded query scheme is equivalent to solving a linear programming problem.


Upon receiving the query $Q_t\subseteq\cN$, the server generates a corresponding answer $A_t=W_{Q_t}\subseteq W_\cN$. The length of the answer can be written as
\begin{equation*}
	\ell\left(Q_t\right)=|Q_t|\,L,
\end{equation*}
where $L$ is the length of a message. Therefore, the average length $\ell_t$ is 
\begin{equation}
\label{eq:scheme-download}
	\ell_t = \mathbb{E}\left[|Q_t|\right]\,L.
\end{equation}

Next, we describe how to construct the query $Q_t$ for each time $t$.
The query $Q_t$ is a probabilistic function of the current request
$X_t$ and $U_t=(X_\tau,X_{t+1})$(c.f.\eqref{eq:def-u}). 
Therefore, the encoding of the query $Q_t$ can be equivalently denoted by the probability distribution $\wb{q_t|x_t,u_t}$, where $x_t \in \cN$, $u_t=(x_{\tau},x_{t+1}) \in \cN^2$ and $q_t \in \sP\left(\cN\right)$. In other words, given $x_t \in \cN$ and $u_t=(x_{\tau},x_{t+1}) \in \cN^2$, the user will send $q_t \in \sP\left(\cN\right)$ with probability $\wb{q_t|x_t,u_t}$.

Since $A_t=W_{Q_t}$, if $X_t \in Q_t$, then the retrieved answer contains the desired message $W_{X_t}$. Therefore, if 
\begin{equation}
\label{eq:scheme-decode}
	\pb{q_t,x_t|u_t} = 0,\quad \forall x_t \notin q_t,
\end{equation}
then decodability is guaranteed. Note that $\pb{q_t,x_t|u_t}$ can be written as 
\begin{equation*}
	\pb{q_t,x_t|u_t} = \pb{x_t|u_t} \wb{q_t|x_t,u_t},
\end{equation*}
where $\pb{x_t|u_t}$ is given by the Markov chain, so $\pb{q_t,x_t|u_t}$ is completely determined by $\wb{q_t|x_t,u_t}$.

To guarantee the privacy(c.f.\eqref{eq:privacy}), we introduce the following lemma. It states that if we design the encoding function $\wb{q_t|x_t,u_t}$ such that 
\begin{equation}
\label{eq:scheme-privacy}
	\pb{q_t|u_t} = \pb{q_t},\quad \forall u_t \in \cN^2,  q_t \in \sP\left(\cN\right),
\end{equation}
for all $t \in \mathbb{N}$, then the scheme satisfies the required privacy constraint \eqref{eq:privacy}.
\begin{lemma}
\label{Lemma:privacy}
	If $Q_i$ is a probabilistic function of $U_i$ and $X_i$, and $Q_i$ is independent of $U_i$ for $i =0,1,\ldots,t$, then $Q_{[t]}$ is independent of $X_{\cB_t}$, where $\cB_t = \{i:  i \leq t, F_i=\text{ON}\} \cup \{i: i \geq t+1\}$.
\end{lemma}
\begin{proof}
	See Appendix~\ref{Appendix:proof-privacy}.
\end{proof}

Since the download cost of the scheme is as given in \eqref{eq:scheme-download}, \Ie
$
	\ell_t = \mathbb{E}\left[|Q_t|\right]\,L,	
$
and we desire a scheme with low download cost (high rate), we would like to design an encoding function $\wb{q_t|x_t,u_t}$ that 
minimizes $\mathbb{E}\left[|Q_t|\right]$. 

Hence, it remains to design the distribution $\wb{q_t|x_t,u_t}$ that minimizes $\mathbb{E}\left[|Q_t|\right]$ under the  constraints 
\eqref{eq:scheme-decode} and \eqref{eq:scheme-privacy}.
As such, any feasible solution to the following optimization problem corresponds to an admissible encoding function $\wb{q_t|x_t,u_t}$ as desired. 
\begin{equation}
\label{eq:scheme-LP}
    \begin{aligned}
    & \underset{\wb{q_t|x_t,u_t}}{\text{minimize}}
    & & \mathbb{E}\left[|Q_t|\right] = \sum_{q_t} \pb{q_t}|q_t| &\\
    & \text{subject to}
    & & \pb{x_t,q_t|u_t} =0, \ x_t \notin q_t,  & \\
    & & & \pb{q_t|u_t} = \pb{q_t}. & \\
    \end{aligned}
\end{equation}
Note that the problem is always feasible, as
\begin{equation}
 	\wb{q_t=\cN|x_t,u_t} = 1, ~ \forall u_t, x_t,
 \end{equation} 
is a feasible solution to \eqref{eq:scheme-LP}.

One may also notice that if we treat each probability $\wb{q_t|x_t,u_t}$ for $x_t \in \cN$, $u_t \in \cN^2$ and $q_t \in \sP\left(\cN\right)$ as a decision variable, then both the objective function and two constraints are linear, and hence the optimization problem \eqref{eq:scheme-LP} is indeed a linear programming instance. However, this linear programming problem has $n^3\,2^n$ variables and $n\,2^{n-1}+n^2\,2^n$ constraints.
The scale of the problem is intractable in complexity with any generic linear programming solver. For example, using the techniques presented in \cite{Vaidya}, the complexity of this linear programming is $\cO\left(\left(n^2\,2^{n}\right)^{2.5}\right)$.
This makes the numerical solution impossible when $n$ is large.

Therefore, in the following section, we present a polynomial time algorithm that gives a feasible solution that might not always be optimal.

\section{Efficient ON-OFF Privacy Query Scheme}
\label{section:algorithm_achievable}
Instead of attempting to solve the linear programming problem \eqref{eq:scheme-LP} numerically, we are going to identify a feasible solution $w^{\ast}\left(q_t|x_t,u_t\right)$ to the problem \emph{efficiently}, and bound the objective $\mathbb{E}\left[|Q_t|\right]$ \emph{analytically}, \Ie a feasible solution attains an objective such that
\begin{equation}
\label{eq:LP-feasible}
    \mathbb{E}\left[|Q_t^{\ast}|\right] \leq 1/R^{I}_t =  \sum_{i=1}^{n} i \, \theta_i(t), 
\end{equation}
which means there exists a scheme such that the download cost $\ell_t$ is less than or equal to $L/R^{I}_t$, or $R^{I}_t$ is achievable.

A key observation on \eqref{eq:scheme-LP} is that any tractable solution $\wb{q_t|x_t,u_t}$ must be sparse, \Ie a few non-zero valued probabilities $\wb{q_t|x_t,u_t}$ for $x_t \in \cN$, $u_t \in \cN^2$ and $q_t \in \sP{\left(\cN\right)}$. Otherwise, simply initializing or outputting the solution $\wb{q_t|x_t,u_t}$ introduces an exponential overhead in complexity. 
This observation motivates our algorithm, which admits a sparse $\wb{q_t|x_t,u_t}$. 

Since the time index $t$ will be clear from context, in the sequel we drop it from the subscripts. For any given $\pb{x|u}$, we recall the optimization problem we are interested in,
\begin{equation}
\label{eq:scheme-LP2}
    \begin{aligned}
    & \underset{\wb{q|x,u}}{\text{minimize}}
    & & \mathbb{E}\left[|Q|\right] = \sum_{q} \pb{q}|q| &\\
    & \text{subject to}
    & & \pb{x,q|u} =0, \ x \notin q,  & \\
    & & & \pb{q|u} = \pb{q}, & \\
    \end{aligned}
\end{equation}
where $x \in \cN$, $u \in \cN^2$ and $q \in \sP{\left(\cN\right)}$.

Instead of finding a feasible solution to \eqref{eq:scheme-LP2} directly, we introduce an auxiliary random variable $Z$. Let $Z$ be a multiset~$\left(\cN,f \right)$, where $\cN$ is the ground set and $f$ is the multiplicity function. The cardinality of the multiset $Z$ is the summation of multiplicities of all its element, \Ie
\begin{equation*}
    |Z|= \sum_{x \in \cN} f(x).
\end{equation*}
Let $\cZ$ be the collection of all multisets such that cardinality is bounded by $n$, \Ie 
\begin{equation}
	\cZ = \left\{Z: Z \in \left(\cN,f \right), |Z| \leq n \right\}.
\end{equation}

Then for any given $\pb{x|u}$, we can define an alternative optimization problem:
\begin{equation}
\label{eq:scheme-LP-multiset}
    \begin{aligned}
    & \underset{\wb{z|x,u}}{\text{minimize}}
    & & \mathbb{E}\left[|Z|\right] = \sum_{z} \pb{z}|z| &\\
    & \text{subject to}
    & & \pb{x,z|u} =0, \ x \notin z,  & \\
    & & & \pb{z|u} = \pb{z}, & \\
    \end{aligned}
\end{equation}
where $x \in \cN$, $u \in \cN^2$ and $z \in \cZ$.

One can easily check that any feasible solution to \eqref{eq:scheme-LP-multiset} can be easily transformed to be a feasible solution to \eqref{eq:scheme-LP2} by simply letting $Q=\text{Set}(Z)$, \Ie forcing the multiplicity of elements in $Z$ to be $1$. Moreover, the corresponding solution to \eqref{eq:scheme-LP2} attains a better objective value, \Ie if $\wb{z|x,u}$ is a   feasible solution to \eqref{eq:scheme-LP-multiset} and $\wb{q|x,u}$ is a feasible solution to \eqref{eq:scheme-LP2}, then $\mathbb{E}\left[|Q|\right] \leq \mathbb{E}\left[|Z|\right]$, where $\mathbb{E}\left[|Z|\right]$ is the objective value attained by $\wb{z|x,u}$ and $\mathbb{E}\left[|Q|\right]$ is the objective value attained by $\wb{q|x,u}$, respectively. 
Therefore, we will study the feasible region of \eqref{eq:scheme-LP-multiset} instead. In particular, we will find
a feasible solution $w^{\ast}\left(z|x,u\right)$ such that 
\begin{equation}
	\mathbb{E}\left[|Z^{\ast}|\right] = \sum_{i=1}^{n} i \, \theta_i. 
\end{equation}
Then there exists a corresponding feasible solution $w^{\ast}\left(q|x,u\right)$, by simply letting $Q=\text{Set}(Z)$, to the original problem \eqref{eq:scheme-LP2} such that 
\begin{equation}
	\mathbb{E}\left[|Q^{\ast}|\right] \leq \sum_{i=1}^{n} i \, \theta_i,
\end{equation}
which is the same as \eqref{eq:LP-feasible} and is to be proved.

In the remainder of this section, we start by  describing the algorithm in Subsection~\ref{subsection:algorithm}. We then analyze its complexity in Subsection~\ref{subsection:alg_complexity}, and finally in Subsection~\ref{subsection:alg_verification}, we verify that the algorithm outputs a feasible solution as desired.

\subsection{Algorithm Description}
\label{subsection:algorithm}
In this section, we describe the algorithm to construct a feasible solution $\wb{z|x,u}$ to \eqref{eq:scheme-LP-multiset}, \Ie for any given distribution $\pb{x|u}$, we will give a constructive proof of some $Z$, satisfying that
\begin{equation}
\label{eq:alg-decode}
    \pb{z,x|u}=0,~\forall x \notin z,
\end{equation}
and 
\begin{equation}
\label{eq:alg-privacy}
   \pb{z|u}=\pb{z|u'},~\forall z \in \cZ~\text{and}~u,u' \in \cN^2.
\end{equation}
In particular, we will show that the feasible solution $\wb{z|x,u}$ gives  
\begin{equation}
\label{eq:alg-rate}
	\pb{|Z|=\ell} =  \theta_\ell, ~\ell=1,\ldots, n.
\end{equation}
Note that $\theta_{i} \geq 0$ for all $i=1,\ldots,n$ by the definition \eqref{eq:def-theta}, which is stated in the following proposition. 
\begin{proposition}
\label{proposition:lamma-less-one}
For any given Markov chain and time index $t$,
$
	\theta_{i} \geq 0, 
$
for $i=1,\ldots,n$.
\end{proposition}
\begin{proof}
	See Appendix~\ref{appendix:proposition-bound-one}.
\end{proof}

One can see that the objective value attained by this feasible solution is 
\begin{equation*}
	\mathbb{E}\left[|Z|\right] = \sum_{i=1}^{n} i \, \theta_i.
\end{equation*}

Before describing the steps of the algorithm we give an intuitive explanation and overview of the algorithm.
In order to minimize $\mathbb{E}\left[|Z|\right]$, we would like to construct some $\wb{z|x,u}$ that makes the probability $\pb{|Z|=\ell}$ larger for smaller $\ell$, \Ie a greedy-like algorithmic approach is appealing. As a result of the two constraints \eqref{eq:alg-decode} and \eqref{eq:alg-privacy}, one can easily check that the maximum value of $\pb{|Z|=1}$ is $\theta_1$, and the solution gives 
\[\pb{z,x|u} = \min_{u' \in \cN^2} \pb{x|u'},~\forall u~\text{and}~z=x.\]
We would like to keep this greedy manner to manage the probabilities $\pb{z,x|u}$ for $|z|=2,\ldots,n$. However, when $|z| \geq 2$, it becomes more complicated. For instance, when $|z|=2$, one of the two elements of the set $z$ has to be $x$, in order to satisfy \eqref{eq:alg-decode}, which corresponds to the decodability constraint. Roughly speaking, we aim to use the second element of $z$ to obfuscate each $x$ with another $x'$ in order to satisfy \eqref{eq:alg-privacy}, which corresponds to the privacy constraint. The challenging part of this algorithm is this choice of $x'$, and the corresponding probability $p(z,x|u)$, where $z=\{x,x'\}$.

The following algorithm, consisting of five main steps, rigorously describes how we design this obfuscation. In Step 1, we calculate preliminaries from the given probability distribution $\pb{x|u}$ and initialize the algorithm. In Step 2, we describe how to properly obfuscate each $x$ with the other $\ell-1$ elements for a given $\ell$, and in Step 3, we describe how to design a common obfuscation (obtain some common sets $z$ of cardinality $\ell$ and some proper values) for all $x \in \cN$ simultaneously. Then, in Step 4, we augment the configurations to the initialized variables, and finally in Step 5, we output the configurations and the values.  Details are given as follows:

\noindent
\textbf{$\bullet$ Step 1: Preliminaries}

For any given distribution $\pb{x|u}$, by sorting $\pb{x|u}$ for each $x \in \cN$, we can easily obtain parameters 
\[\left\{u_{x,i}: x \in \cN, i = 1,\ldots, m \right\},\]
where $u_{x,i}$ is as defined in \eqref{eq:def-order} and $m=|\cN^2|=n^2$. For notational simplicity, let 
\begin{equation*}
	\lambda_{x,i}= \pb{X_t=x|U_t=u_{x,i}}.
\end{equation*}

Let $M$ be an auxiliary $m \times n$ matrix determined by the given $\pb{x|u}$. In particular, we initialize $M$ by 
\begin{equation}
\label{eq:alg-initialize}
     M_{i,j}  =  \max\left\{\pb{X=j|U=i}- \lambda_{j,n-1}, 0\right\}.
\end{equation}
for $i=1,\ldots,m$, and $j=1,\ldots,n$. This matrix will be updated during the following procedure. For the ease of notation, let $M^{-}_{i,j} = a$ denote $M_{i,j} = M_{i,j} - a$, \Ie subtracting $a$ from $M_{i,j}$.

For $\ell=1,\ldots,n-1$ and $x=1,\ldots,n$, we access to $\{u_{x,i}: i =1, \ldots, \ell-1\}$.
For ease of notation, let 
\begin{equation*}
	\cU_{\ell,x}^{-}=\left\{u_{x,i}: i =1,\ldots,\ell-1\right\},
\end{equation*}
and
\begin{equation*}
	\cU_{\ell,x}^{+}=\left\{u_{x,i}: i =\ell,\ldots,m\right\}.
\end{equation*}

\noindent
\textbf{$\bullet$ Step 2: }

For each $i$, or precisely $u_{x,i}$, we choose a collection of pairs 
\begin{equation}
\label{eq:alg-row-compensate}
	I_{\ell,x,i} \times V_{\ell,x,i} = \left\{\left(e_{\ell,x,i,j},v_{\ell,x,i,j} \right):j=1,2,\ldots,c_{\ell,x,i} \right\}   
\end{equation}
such that 
\begin{equation}
\label{eq:alg-row-budget}
	0 \leq v_{\ell,x,i,j} \leq M_{u_{x,i},e_{\ell,x,i,j}},  
\end{equation}
and
\begin{equation}
\label{eq:alg-row-sum}
	\begin{aligned}
		\sum_{j=1}^{c_{\ell,x,i}} v_{\ell,x,i,j} = \lambda_{x,\ell} - \lambda_{x,\ell-1}. 
	\end{aligned}
\end{equation}
where $e_{\ell,x,i,j}$ for $j=1,\ldots,c_{\ell,x,i}$ are distinct indices belonging to $\{1,\ldots,n\}$, and clearly we have $c_{\ell,x,i} \leq n$.  

Then, we update the matrix $M$ by
\begin{equation}
\label{eq:alg-update}
	 M_{u_{x,i},e_{\ell,x,i,j}}^{-} = v_{\ell,x,i,j},
\end{equation}
for all $e_{\ell,x,i,j} \in I_{\ell,x,i}$. We slightly abuse the notation here by using the same notation $M$ to denote the matrix at different points.  Nevertheless, the underlying $\ell$, $x$ and $i$ we are dealing with will be clear from context.

Roughly speaking, we extract  $\lambda_{x,\ell} - \lambda_{x,\ell-1}$ from the $u_{x,i}$-th row of the non-negative matrix $M$ for given $\ell$ and $x$, where $e_{\ell,x,i,j}$ and $v_{\ell,x,i,j}$ specify the column indices and values extracted from each position of $u_{x,i}$-th row. 
The matrix $M$ is always non-negative during the update from \eqref{eq:alg-row-budget} and \eqref{eq:alg-update}, so the existence of such a collection of $I_{\ell,x,i} \times V_{\ell,x,i}$ can be guaranteed if the summation of the $u_{x,i}$-th row of the initialized matrix $M$(c.f.\eqref{eq:alg-initialize}) is greater than or equal to the summation of the subtracted values (the right-hand side of \eqref{eq:alg-existence}) for all $x$ and $\ell$ during the process, which is given by the following proposition. 
\begin{proposition}
\label{prop:alg-existence}
For any $u=1,\ldots,m$,
 \begin{equation}
\label{eq:alg-existence}
\begin{aligned}
	 & \sum_{x=1}^{n} \max\left\{\pb{X=x|U=u}- \lambda_{x,n-1}, 0\right\}  \\
	 & ~~~~ \geq \sum_{\ell=1}^{n-1} \sum_{x:u \in \cU_{\ell,x}^{-}} \left( \lambda_{x,\ell} - \lambda_{x,\ell-1}\right).
\end{aligned}
\end{equation}
\end{proposition}
\begin{proof}
    See Appendix~\ref{Appendix:proof-alg-existence}.
\end{proof}

\noindent
\textbf{$\bullet$ Step 3:}

For fixed $\ell$ and $x$, after finishing the above process for all $i=1,\ldots,\ell-1$, we obtain $I_{\ell,x,i}$ and $V_{\ell,x,i}$ for $i= 1,\ldots,\ell-1$. Provided  $I_{\ell,x,i}$ and $V_{\ell,x,i}$ for $i= 1,\ldots,\ell-1$, we pick a collection of pairs 
\begin{equation*}
	\left\{(\zeta_{\ell,x,k},\nu_{\ell,x,k}): k=1,2,\ldots, c_{\ell,x} \right\}
\end{equation*} 
such that 
\begin{equation*}
	\zeta_{\ell,x,k} \in I_{\ell,x,1} \times I_{\ell,x,2} \times \cdots \times I_{\ell,x,\ell-1},
\end{equation*}
and 
\begin{equation}
\label{eq:alg-multiset-value-zx}
	\sum_{k: \zeta_{\ell,x,k}(i)=e_{\ell,x,i,j}} \nu_{\ell,x,k} = v_{\ell,x,i,j}, 
\end{equation}
for all $i = 1,\ldots, \ell-1$ and  $j = 1,\ldots, c_{\ell,x,i}$, where $\zeta_{\ell,x,k}(i)$ is the $i$-th element of $\zeta_{\ell,x,k}$, \Ie $\zeta_{\ell,x,k}(i) \in I_{\ell,x,i}$.

A simple deterministic approach of picking such a collection of $(\zeta_{\ell,x,k},\nu_{\ell,x,k})$ can be basically illustrated by Figure~\ref{fig:alg-boundary}. Roughly speaking, there is a buffer tracking the front of the sets $V_{\ell,x,i}$ for $i=1,\ldots,\ell-1$. Each time, the buffer pushes the minimal value among them \Ie $\nu_{\ell,x,k}$, minus the value from the front, and adds one more value from the same set $V_{\ell,x,i}$ which has been pushed out. The corresponding positions of values in the buffer form the set $\zeta_{\ell,x,k}$. As such, we can easily see that 
\begin{equation}
\label{eq:alg-sum-nu}
\begin{aligned}
	\sum_{k=1}^{c_{\ell,x}} \nu_{\ell,x,k} = \lambda_{x,\ell} - \lambda_{x,\ell-1}.
\end{aligned}
\end{equation}
Also, one can easily check that this process returns 
\begin{equation}
\label{eq:size-clx}
	c_{\ell,x} = \sum_{i=1}^{\ell-1} c_{\ell,x,i} \leq n(\ell-1).
\end{equation}
\begin{figure}[t]
    \centering
    \scalebox{0.8}{
         \begin{tikzpicture}
  \draw[ultra thick,fill=bleudefrance2] (0,5) rectangle (2,6);
  \draw[ultra thick,fill=darkgreen2] (2,5) rectangle (5,6);
  \draw[ultra thick,fill=magenta2] (5,5) rectangle (8,6);
  \node at (1,5.5) {$v_{\ell,x,1,1}$};
  \node at (3.5,5.5) {$v_{\ell,x,1,2}$};
  \node at (6.5,5.5) {$v_{\ell,x,1,3}$};
  \draw[ultra thick,fill=darkgreen2] (0,3.5) rectangle (4,4.5);
  \draw[ultra thick,fill=magenta2] (4,3.5) rectangle (8,4.5);
  \node at (2,4) {$v_{\ell,x,2,1}$};
  \node at (6,4) {$v_{\ell,x,2,2}$};
  \draw[ultra thick,fill=bleudefrance2] (0,1) rectangle (6,2);
  \draw[ultra thick,fill=darkgreen2] (6,1) rectangle (8,2);
  \node at (3,1.5) {$\cdots$};
  \node at (7,1.5) {$v_{\ell,x,i,c_{\ell,x,i}}$};
  \node at (7.8,2.75) [rotate=90] {\Large $\dots$};
  \draw[red,dashed,thick] (2,6.25) -- (2,1);
  \draw[red,dashed,thick] (4,6.25) -- (4,1);
  \draw[red,dashed,thick] (5,6.25) -- (5,1);
  \draw[red,dashed,thick] (6,6.25) -- (6,1);
  \draw[thick,|-|] (0,6.25) -- (2,6.25);
  \node at (1,6.5) {$\nu_{\ell,x,1}$};
  \draw[thick,|-|] (2,6.25) -- (4,6.25);
  \node at (3,6.5) {$\nu_{\ell,x,2}$};
  \draw[thick,|-|] (4,6.25) -- (5,6.25);
  \node at (4.5,6.5) {$\nu_{\ell,x,3}$};
  \draw[thick,|-|] (5,6.25) -- (6,6.25);
  \node at (5.5,6.5) {$\cdots$};
  \draw[thick,|-|] (6,6.25) -- (8,6.25);
  \node at (7,6.5) {$\nu_{\ell,x,c_{\ell,x}}$};
\end{tikzpicture}    %
         }

\caption{The rows represents $V_{\ell,x,1},\ldots,V_{\ell,x,\ell-1}$ for given $\ell$ and $x$. Each block represents an element $v_{\ell,x,i,j}$ in the set $V_{\ell,x,i}$, where $j=1,\ldots,c_{\ell,x,i}$. Each $\nu_{\ell,x,k}$ can be chosen to be the value of the difference between two consecutive boundaries of blocks, \Eg $\nu_{\ell,x,1} = v_{\ell,x,1,1}$ and $\nu_{\ell,x,2} = v_{\ell,x,2,1} - v_{\ell,x,1,1}$ etc. The corresponding $\zeta_{\ell,x,k}$ can be chosen to be  $\zeta_{\ell,x,1} = \left(e_{\ell,x,1,1}, e_{\ell,x,2,1}, \cdots \right)$ and  $\zeta_{\ell,x,2} = \left(e_{\ell,x,1,2}, e_{\ell,x,2,1}, \cdots \right)$ etc.}
\label{fig:alg-boundary}
\end{figure}
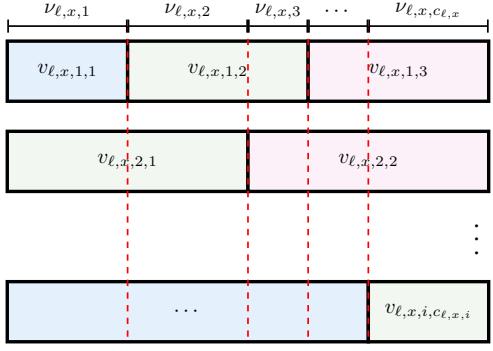

\noindent
\textbf{$\bullet$ Step 4: Augment}

For each $k=1,\ldots,c_{\ell,x}$, let 
\begin{equation}
\label{eq:alg-def-z}
	z_{\ell,x,k}=\{\zeta_{\ell,x,k},x\},
\end{equation}
and 
\begin{equation}
\label{eq:argument-seperate}
	\begin{aligned}
	   & \cF_{\ell,x,k} = \left\{\left(\bar{z},\bar{x},\bar{u}\right):\bar{z}=z_{\ell,x,k},\bar{x}=x,\bar{u} \in \cU_{\ell,x}^{+} \right\} \bigcup \\ 
	   &  \left\{\left(\bar{z},\bar{x},\bar{u}\right): \bar{z}=z_{\ell,x,k}, \bar{x}=\zeta_{\ell,x,k}(i), \bar{u} = u_{x,i} \in \cU_{\ell,x}^{-} \right\}.
	\end{aligned}
\end{equation}

For each $\ell$ and $x$, we can obtain $\cF_{\ell,x,k}$ and  $\nu_{\ell,x,k}$ for $k=1,\ldots,c_{\ell,x}$. 
The tuple $\left(\bar{z},\bar{x},\bar{u}\right)$ in $\cF_{\ell,x,k}$ is indeed the non-zero valued position and $\nu_{\ell,x,k}$ is the value that we will assign to the probability $\pb{\bar{z},\bar{x}|\bar{u}}$. 
However, since there may exist duplicated tuples in $\cF_{\ell,x,k}$ for different $x$, we augment the value $\nu_{\ell,x,k}$ corresponding to the same tuple $\left(\bar{z},\bar{x},\bar{u}\right)$, \Ie 
\begin{equation}
\label{eq:argument-l}
	\cF_{\ell} =\bigcup_{x=1}^{n} \bigcup_{k=1}^{c_{\ell,x}} \cF_{\ell,x,k},
\end{equation}
and for any $\left(\bar{z},\bar{x},\bar{u}\right) \in \cF_{\ell}$, 
\begin{equation}
\label{eq:assign-l}
g\left(\bar{z},\bar{x},\bar{u}\right)
	= \sum_{x=1}^{n} \sum_{k:\left(\bar{z},\bar{x},\bar{u}\right) \in \cF_{\ell,x,k}}  \nu_{\ell,x,k}.
\end{equation}

After obtaining $\cF_\ell$ for $\ell=1,\ldots,n-1$, for $\ell = n$, let
\begin{equation}
\label{eq:argument-n}
	\cF_{n}= \left\{\left(\bar{z},\bar{x},\bar{u}\right): \bar{z} = \cN, M_{\bar{u},\bar{x}}>0 \right\},
\end{equation}
and
\begin{equation}
\label{eq:assign-n}
	g(\bar{z},\bar{x},\bar{u})= M_{\bar{u},\bar{x}},
\end{equation}
for any $\left(\bar{z},\bar{x},\bar{u}\right) \in \cF_{n}$. 

\noindent
\textbf{$\bullet$ Step 5: Output}

The output of the algorithm is $\left\{\cF,g(\cF)\right\}$, where 
\begin{equation*}
	\cF=\{\cF_\ell: \ell =1,\ldots,n\}.
\end{equation*}
stores the non-zero valued positions of an admissible distribution $\pb{z,x|u}$ for $x \in \cN$, $u \in \cN^2$ and $z \in \cZ$, and $g(\cF)$ stores the corresponding probabilities.

\subsection{Complexity}
\label{subsection:alg_complexity}
For the sake of completeness, we discuss the complexity of the algorithm.  
As said, the bottleneck is to represent the solution $\wb{z|x,u}$ for $z \in \cZ$, $x \in \cN$ and $u \in \cN^2$, which has exponential number of values, so the complexity is indeed dominated by the size of $\cF$, \Ie the non-zero valued positions of the output distribution $\pb{z,x|u}$. 

It is notable that 
$|\cF_{\ell,x,k}|=m$  and $c_{\ell,x} \leq n^2$ from \eqref{eq:size-clx}, so  we have 
\begin{equation}
	|\cF| \leq \sum_{\ell=1}^{n} \sum_{x=1}^{n} \sum_{k=1}^{c_{\ell,x}} |\cF_{\ell,x,k}| \leq m n^4 = n^6,  
\end{equation}
\Ie the complexity of the algorithm is $\cO(n^6)$.

The purpose of the complexity analysis here is to justify that the proposed algorithm is 
with $\text{poly}(n)$ complexity. One may possibly reduce the complexity by orders by utilizing some data structures, which is beyond the interest of this paper.

\subsection{Algorithm Verification}
\label{subsection:alg_verification}
In this subsection, we will verify the algorithm, \Ie we will prove that it outputs a distribution $\pb{z,x|u}$ satisfying  \eqref{eq:alg-decode}, \eqref{eq:alg-privacy} and \eqref{eq:alg-rate} for any given distribution $\pb{x|u}$.

First, we show that the algorithm described in \ref{subsection:algorithm} outputs a distribution $\pb{z,x|u}$ satisfying \eqref{eq:alg-decode} and \eqref{eq:alg-privacy} for any given distribution $\pb{x|u}$. 
\begin{proposition}
\label{proposition:alg-justification}
	For any given $\pb{x|u}$ for $u \in \cN^2$ and $x \in \cN$, $\left\{\cF, g(\cF) \right\}$ returns non-zero valued positions and values of some distribution $\pb{z,x|u}$ such that $\pb{z,x|u}=0$ for all $x \notin z$ and $\pb{z|u}=\pb{z|u'}$ for all $z \in \cZ$ and $u,u' \in \cN^2$.
\end{proposition} 
\begin{proof}

As claimed, $\cF$ and $g\left(\cF\right)$ store the non-zero valued positions and values of $\pb{z,x|u}$, so it is equivalent for us to show that
\begin{enumerate}
	\item For any $\left(\bar{z},\bar{x},\bar{u}\right) \in \cF$, we have 
		\begin{equation}
			\bar{x} \in \bar{z}.
		\end{equation}
	\item \label{item:two} For any given $\bar{z}$, $\bar{u}$ and $\bar{u}'$, we have
		\begin{equation}
			\sum_{x:\left(\bar{z},x,\bar{u}\right) \in \cF} g\left(\bar{z},x,\bar{u}\right) = \sum_{x:\left(\bar{z},x,\bar{u}'\right) \in \cF} g\left(\bar{z},x,\bar{u}'\right).
		\end{equation}
	\item For any given $\bar{x}$ and $\bar{u}$, we have	
		\begin{equation}
			\sum_{z:\left(z,\bar{x},\bar{u}\right) \in \cF} g\left(z,\bar{x},\bar{u}\right) = \pb{\bar{x}|\bar{u}}.
		\end{equation}		
\end{enumerate}
Details can be found in Appendix~\ref{Appendix:proposition-justify}.
\end{proof}

Next, we show that the algorithm described in \ref{subsection:algorithm} returns $\pb{z,x|u}$ satisfying  \eqref{eq:alg-rate}.
\begin{proposition}
\label{proposition:alg-performance}
	For any given $\pb{x|u}$ for $u \in \cN^2$ and $x \in \cN$,
the algorithm returns some distribution $\pb{z,x|u}$ such that 
\begin{equation}
	\pb{|Z|=\ell} = \theta_{\ell}, ~\ell=1,\ldots,n.
\end{equation}
\end{proposition} 
\begin{proof}
	See	Appendix~\ref{Appendix:proposition-perform}
\end{proof}

\section{An outer bound}
\label{section:converse}

In this section, we will show that any ON-OFF privacy scheme must satisfy $R_t \leq R_t^{O}$. 

First, we define an auxiliary random variable $Y_t$ taking values in $\sP\left(\cN\right)$ based on the decodability of the subset of messages. Specifically, let $Y_t$ be a function of $Q_t$ such that $Y_t = \cD$ for $\cD \in \sP\left(\cN\right)$ if the user may decode the messages $W_{\cD}$ but not any message $W_{i}$ for $i \in \cN\backslash \cD$ from the answer $A_t$. 
Roughly speaking, $Y_t$ represents the capability of decoding messages from the query $Q_t$. Note that since the query $Q_t$ and messages $W_{\cN}$ are independent, the decodability of any message is known by the server only through $Q_t$, that is, $Y_t$ is a function of $Q_t$. 
In this way, the alphabet $\cQ$ (may be infinite if the query is coded) of the query is partitioned into $2^n$ classes based on the decodability of the subset of the messages. 
Clearly, from the definition of $Y_t$, we have that
the length of the answer $\ell(Q_t)$ satisfies 
\begin{equation*}
	\ell(Q_t) \geq |Y_t|\,L,
\end{equation*}
since the answer $A_t$ is at least of length $|Y_t|\,L$ if the user can decode $|Y_t|$ messages from the answer $A_t$.
Hence, the download cost $\ell_t$ is bounded by 
\begin{equation}
\label{eq:converse-length}
	\ell_t \geq \mathbb{E}\left[|Y_t|\right]\,L.
\end{equation}

Next, we start to reinterpret the privacy and the decodability constraints  in terms of the auxiliary variable $Y_t$. By the definition of $Y_t$, the decodability can be written as 
\begin{equation}
\label{eq:converse-decode}
	\pb{x_t,y_t}=0, \forall x_t \notin y_t,
\end{equation}
where $x_t \in \cN$ and $y_t \in \sP\left(\cN\right)$.

Recall the privacy constraint 
\begin{equation*}
  I \left(X_{\cB_t};Q_{[t]}\right) = 0,
\end{equation*} 
and we must have 
\begin{align*}
	I \left(X_{\cB_t};Q_{[t]}\right)
	 \utag{a}{\geq} I\left(U_{t};Q_t\right) 
	 \utag{b}{\geq} I\left(U_{t};Y_t\right), 
\end{align*}
where \uref{a} follows from $U_t=(X_{\tau},X_{t+1}) \subset X_{\cB_t}$ and \uref{b} follows because $Y_t$ is a function of $Q_t$.

Thus, we can relax the privacy constraint by 
\begin{equation}
\label{eq:converse-privacy}
 I \left(U_{t};Y_t\right) = 0.
\end{equation}

For any given $\pb{x_t|u_t}$, if $Y_t$ takes values in $\sP\left(\cN\right)$ and satisfies \eqref{eq:converse-decode} and \eqref{eq:converse-privacy}, then $\mathbb{E}\left[|Y_t|\right]$ is lower bounded by the following lemma. 
\begin{lemma}
\label{lemma:outer}
	For any random variables $U$, $X$ and $Y$, taking values in the alphabet $\cN^2$, $\cN$ and $\sP\left(\cN\right)$ respectively, if $Y$ is independent of $U$, and $p(x,y|u)=0$ for $x \notin y$, then 
	\begin{equation}
	\label{eq:lemma-expectation-lower}
		\mathbb{E}\left[|Y|\right] \geq \sum_{x \in \cN} \max_{u \in \cN^2} \pb{x|u}. 
	\end{equation}
\end{lemma}
\begin{proof}
	See Appendix~\ref{Appendix:lemma-outer}.
\end{proof}

By substituting \eqref{eq:lemma-expectation-lower} in \eqref{eq:converse-length}, we have 
\begin{equation}
\label{eq:converse-download-lower-bound}
	\ell_t \geq L\,\sum_{x_t \in \cN} \max_{u_t \in \cN^2} \pb{x_t|u_t}.
\end{equation}

Therefore, for any ON-OFF privacy scheme satisfying the decobability and privacy constraint, we know that the download cost is lower bounded by the right-hand side of
\eqref{eq:converse-download-lower-bound}.
In other words, any ON-OFF privacy scheme must satisfy 
\begin{equation*}
	\frac{1}{R_t} \geq \frac{1}{R_t^{O}} = \sum_{x_t \in \cN} \max_{u_t \in \cN^2} \pb{x_t|u_t} = \lambda_m(t).
\end{equation*}

\section{LP Formulation of optimal achievable rate}
\label{section:discussion}
In this section, we present an implicit characterization of the optimal rate, which is formulated by a linear program with an exponential number (in $n$) of variables and constraints. 

As discussed in Section~\ref{section:scheme}, the query design relies on solving the following linear program: 
\begin{equation}
\label{eq:LP-achievable}
    \begin{aligned}
    & \underset{\wb{q|x,u}}{\text{minimize}}
    & & \mathbb{E}\left[|Q|\right] = \sum_{q} \pb{q}|q| &\\
    & \text{subject to}
    & & \pb{x,q|u} =0, \ x \notin q,  & \\
    & & & \pb{q|u} = \pb{q}, & \\
    \end{aligned}
\end{equation}
where $x \in \cN$, $u \in \cN^2$, $q \in \sP{\left(\cN\right)}$ and probabilities $\pb{x|u}$ are given. We know that any feasible solution to the above problem yields an achievable scheme. In other words, the rate $R_t$ is achievable if 
\begin{equation*}
	\frac{1}{R_t} \geq C^{\ast}_1,
\end{equation*}
where $C^{\ast}_1$ is the optimal value to \eqref{eq:LP-achievable}.

On the other hand, one may notice that the key lemma, \Ie Lemma~\ref{lemma:outer}, to show the outer bound, indeed indicates that any achievable scheme must satisfy that 
\begin{equation*}
	\frac{1}{R_t} \geq C^{\ast}_2, 
\end{equation*}
where $C^{\ast}_2$ is the optimal value to the following problem:
\begin{equation}
\label{eq:LP-converse}
    \begin{aligned}
    & \underset{\wb{y|x,u}}{\text{minimize}}
    & & \mathbb{E}\left[|Y|\right] = \sum_{y} \pb{y}|y| &\\
    & \text{subject to}
    & & \pb{x,y|u} =0, \ x \notin y,  & \\
    & & & \pb{y|u} = \pb{y}, & \\
    \end{aligned}
\end{equation}
where $x \in \cN$, $u \in \cN^2$, $y \in \sP{\left(\cN\right)}$ and probabilities $\pb{x|u}$ are given by the Markov chain.

Although problems \eqref{eq:LP-achievable} and \eqref{eq:LP-converse} have different physical meanings, it is easy to see that they have the same optimal value, \Ie $C^{\ast}_1 = C^{\ast}_2$. Therefore, by letting 
$C_t$ be the optimal value to both problems, the achievable region can be fully characterized by
\begin{corollary}
\label{corollary:main}
The rate tuple $\left(R_t:t \in \mathbb{N}\right)$ is achievable if and only if $R_t \leq C_t$.
\end{corollary}

However, it is notable that Corollary~\ref{corollary:main} is an \emph{implicit} characterization, because as we discussed, the exponential blow-up of the number of variables and constraints makes the linear programming problem intractable.

\begin{remark}  [Window size $\omega$]
From our earlier discussion, we know that the feasible region of \eqref{eq:LP-achievable} denotes schemes that only require a window of size $\omega = 1$. 
Although we have assumed that the user knows the future requests within a window of positive size $\omega$, increasing the window size into the future beyond $\omega=1$ does not, in fact, increase the rate. Intuitively, this phenomenon stems from the Markov assumption we use to model the user's requests. If the window size $\omega=0$, \Ie no future requests are known, the privacy defined in \eqref{eq:privacy}  has to be relaxed, and only past requests can be protected, which was studied in \cite{Ye_2020}.

\end{remark}

In this paper, we actually proposed an explicit scheme to the ON-OFF privacy problem by finding a feasible solution that might not be optimal, to \eqref{eq:LP-achievable}, which is of polynomial time complexity. Moreover, we show that our scheme is optimal for some cases in Corollary~\ref{corollary:two} and Corollary~\ref{cor:sym_markov}.

	\section{Conclusion and Future Directions}
	\label{section:conclusion}	
	In this paper, we continue  to look at the problem of turning privacy ON and OFF in an information retrieval setting when the user's interests are correlated over time. We model this correlation by a Markov chain with $n$ states. Our previous work in \cite{Naim_2019} focused on privacy for past interests. Our work in \cite{Ye_2019} studied privacy for the past and the future, albeit for the special case of Markov 
chains with $n=2$ states. In this paper, we generalize the work in \cite{Ye_2019} to Markov chains with $n\geq 2$ states. We give a new achievable scheme with polynomial time complexity and a general upper bound on the achievable rate. We prove the optimality of our scheme for special cases, namely, a family of  symmetric Markov chains, and two-state Markov chains.  
	
	    Future directions of this work   include finding tighter outer bounds on the rate and efficient constructions of ON-OFF privacy schemes that would achieve them. Also, it is worthy to investigate settings  in which the user's requests follow a different model than the Markov chain, or the user's requests and desired privacy status are correlated.

\appendices

\section{Optimality for $n=2$}
\label{Appendix:proof_n=2}
The special case when $n=2$ was first studied in~\cite{Ye_2019}. For $n=2$, the two bounds \eqref{eq:thm-inner} and \eqref{eq:thm-outer} match, \Ie $R_t^{I} = R_t^{O}$. To see this we write $R_t^{I}$ by 
\begin{equation*}
	\frac{1}{R_t^{I}} = \sum_{i=1}^{m} i\,\theta_i(t) = 2 - \lambda_{1}(t).
\end{equation*}
For a given $x_t$, \Eg $x_t=1$, suppose that $u^{\ast} ~=~ \argmin_{u_t} \pb{x_t|u_t}$. Then we can see that, for $\bar{x}_t=2$, $u^{\ast} = \argmax_{u_t} \pb{\bar{x}_t|u_t}$
since 
$
	\pb{x_t|u_t} + \pb{\bar{x}_t|u_t} = 1,
$
for any $u_t$ when $n=2$. 
Thus, we have 
\begin{equation*}
	\min_{u_t} \pb{x_t|u_t} + \max_{u_t} \pb{\bar{x}_t|u_t} = 1,
\end{equation*}
for any $x_t$ and $\bar{x}_t = \cN\backslash\{x_t\}$, which implies that 
\begin{align*}
	\lambda_{1}(t) & = \sum_{x_t=1}^2 \min_{u_t} \pb{x_t|u_t} 
	= \sum_{x_t=1}^2 \left(1 - \max_{u_t} \pb{\bar{x}_t|u_t} \right) \\
	& = 2 - \sum_{\bar{x}_t=1}^2  \max_{u_t} \pb{\bar{x}_t|u_t} 
	= 2 - \lambda_m(t).
\end{align*}

Therefore, we can obtain that 
\begin{equation*}
	\frac{1}{R_t^{I}} =  2 - \lambda_{1}(t)  = \lambda_m(t) = \frac{1}{R_t^{O}}.
\end{equation*}

\section{Proof of Corollary \ref{cor:sym_markov} }
\label{Appendix:proof_sym} 
We first take the transition matrix $P$ to the power of $t$, \Ie
\begin{equation*}
    (P^t)_{i,j}=\begin{cases}
        \frac{(n-1)^{t-1}-(n\alpha-1)^t}{n(n-1)^{t-1}}, &\quad \text{if} \quad i=j,\vspace{10pt}\\
        \frac{(n-1)^t-(n\alpha-1)^t}{n(n-1)^t},&\quad \text{if} \quad i\neq j,
    \end{cases}
\end{equation*}
for all $i,j\in \{1,\ldots,n\}$.

Then, the probabilities 
$\pb{x_t|u_t}$ can be written as
\begin{equation}
\label{eq:explicit_prob}
    \pb{X_t=j|X_\tau = i, X_{t+1}=k}=\frac{P_{j,k} (P^{\delta})_{i,j}}{(P^{\delta+1})_{i,k}},
\end{equation}
where $\delta=t-\tau$ and $i,j \in \{1,\ldots,n\}$. By invoking the symmetry of the given Markov chain, we notice that  
the right-hand side of \eqref{eq:explicit_prob} can only have a few of expressions depending on the choices of $i$, $j$ and $k$, \Ie
\begin{equation}
\label{eq:cor2_proof}
 	\begin{aligned}
 	& \pb{X_t=j|X_\tau = i, X_{t+1}=k} = \\
 	&
 	\begin{cases}
        \sigma_1:=\tfrac{\alpha\left((n-1)^{\delta}+(n\alpha-1)^{\delta}(n-1)\right)}{(n-1)^{\delta}+(n\alpha-1)^{\delta+1}}, & \hspace{-9pt}\text{if} \hspace{4pt} i=j=k,\vspace{5pt}\\
        \sigma_2:=\tfrac{(1-\alpha)\left((n-1)^{\delta}+(n\alpha-1)^{\delta}(n-1)\right)}{(n-1)^{\delta+1}-(n\alpha-1)^{\delta+1}}, & \hspace{-9pt}\text{if} \hspace{4pt} i=j\neq k,\vspace{5pt}\\
        \sigma_3:=\tfrac{\alpha\left((n-1)^{\delta+1}-(n\alpha-1)^{\delta}(n-1)\right)}{(n-1)^{\delta+1}-(n\alpha-1)^{\delta+1}}, & \hspace{-9pt}\text{if} \hspace{4pt} i\neq j= k,\vspace{5pt}\\
        \sigma_4:=\tfrac{(1-\alpha)\left((n-1)^{\delta}-(n\alpha-1)^{\delta}\right)}{\left((n-1)^{\delta}+(n\alpha-1)^{\delta+1}\right)(n-1)}, & \hspace{-9pt}\text{if} \hspace{4pt} i=k\neq j,\vspace{5pt}\\
        \sigma_5:=\tfrac{(1-\alpha)\left((n-1)^{\delta}-(n\alpha-1)^{\delta})\right)}{(n-1)^{\delta+1}-(n\alpha-1)^{\delta+1}}, & \hspace{-9pt}\text{if} \hspace{4pt} i\neq j\neq k.
        \end{cases} 
 	\end{aligned}
 \end{equation}

    By examining $\sigma_1$ to $\sigma_5$ in \eqref{eq:cor2_proof},
    we have $\sigma_1\geq \sigma_3\geq\sigma_2\geq\sigma_5\geq\sigma_4$,
    for $\frac{1}{n} \leq \alpha \leq 1$. 
   
    For a fixed $j\in \{1,\ldots,n\}$, by counting the number of times each condition of \eqref{eq:cor2_proof}, \Eg $i=j=k$, $i=j\neq k$, etc., is satisfied for $i,k \in \cN$,
    we can get the following ordering of $n^2$ probabilities $\pb{X_t=j|X_\tau = i, X_{t+1}=k}$ (for a fixed $j$):
    \begin{align*}
       \hspace{-5pt}\underbrace{\sigma_4\leq \dots\leq \sigma_4}_{n-1} \leq\underbrace{\sigma_5\leq \dots \leq\sigma_5}_{(n-1)(n-2)}
       \geq&\underbrace{\sigma_2\leq \dots \leq\sigma_2}_{n-1}\leq
       \nonumber\\
       &\underbrace{\sigma_3\leq \dots \leq\sigma_3}_{n-1} 
       \leq\sigma_1.
    \end{align*}
 
Due to the symmetry, this ordering remains the same for all $j\in \{1,\ldots,n\}$. Given this ordering for any fixed $j$, we can check  
    \begin{align*}
        \frac{1}{R_t^O} & =\sum_{x \in \cN} \pb{X_t=x|U_t=u_{x,n^2}}=\sum_{x\in\cN} \sigma_1=n \sigma_1,
    \end{align*}
where $\pb{X_t=x|U_t=u_{x,n^2}}$ is defined in \eqref{eq:def-order}.

Also, from 
\eqref{eq:def-theta}, we can check that $\theta_1= n\sigma_4$, $\theta_i=0$ for $i=2,\dots,n-1$, and $\theta_n=1-n\sigma_4$, so
  we have
    \begin{align*}
        \frac{1}{R_t^I}&=\sum_{i=1}^n i\theta_i= n\sigma_4 + n - n^2\sigma_4.
    \end{align*}

By substituting the expression of $\sigma_1$ and $\sigma_4$ defined in \eqref{eq:cor2_proof}, one can verify that $n\sigma_4 + n - n^2\sigma_4 = n\sigma_1$, which implies in ${R_t^I} ={R_t^O}$.
 This completes the proof of Corollary~\ref{cor:sym_markov}.

\begin{remark}
When $0 \leq \alpha < \frac{1}{n}$, we may follow the same steps as we did but divide the discussion into two cases: $\delta = t-\tau$ is even or odd.
When $\delta$ is even, we have $\sigma_2\geq\sigma_4\geq\sigma_5\geq\sigma_1\geq\sigma_3$, and
    \begin{equation*}
        \frac{1}{R_t^O}=n\sigma_2\leq \frac{1}{R_t^I}=n\sigma_3+n-n^2\sigma_3.
    \end{equation*}
Similarly when $\delta$ is odd, we have $\sigma_5\geq\sigma_4\geq\sigma_2\geq\sigma_3\geq\sigma_1$, and
    \begin{equation*}
        \frac{1}{R_t^O}=n\sigma_5\leq \frac{1}{R_t^I}=\sigma_3(2n-n^2)-n\sigma_1+n.
    \end{equation*}
In both cases, we can see a gap between $R_t^O$ and $R_t^I$, which is as per our observation in
Example~\ref{ex:symmetric}.
\end{remark}


\section{Proof of Lemma~\ref{Lemma:privacy}}
\label{Appendix:proof-privacy}
First, let us recall that $\cB_t = \{i: i \leq t, F_i=\text{ON}\} \cup \{i: i \geq t+1\}$ and $U_t=\left(X_{\tau},X_{t+1}\right)$ where $\tau= \max\{i: i \leq t, F_i=\text{ON}\}$. 

We prove the statement by induction on $t$. Consider the base case $t=0$. From the assumption $F_0 = \text{ON}$ (assumption of this paper), we know that $U_0 = \{X_0, X_1\}$ and $\cB_0 = \{i:  i = 0,1,\ldots\}$. If $Q_0$ is a stochastic function of $X_0$ and $X_1$, and $Q_0$ is independent of $X_0$ and $X_1$, then we have 
\begin{align*}
 	I\left(Q_0;X_{\cB_0}\right) 
 	& = I\left(Q_0;X_{0},X_1\right) + I\left(Q_0;X_{\cB_0}|X_1, X_0\right) = 0,
\end{align*} 
\Ie $Q_0$ is independent of $X_{\cB_0}$. The last equality follows because $Q_0$ is independent of $X_0$ and $X_1$, and $Q_0$ is a stochastic function of $X_0$ and $X_1$.

Now, we start the inductive step. Assume that the statement is true for some $t-1$, \Ie if $Q_i$ is a stochastic function of $U_i$ and $X_i$, and $Q_i$ is independent of $U_i$ for $i =0,1,\ldots,t-1$, then $Q_{[t-1]}$ is independent of $X_{\cB_{t-1}}$. 

Next, for the case $t$, if $Q_i$ is a stochastic function of $U_i$ and $X_i$, and $Q_i$ is independent of $U_i$ for $i =0,1,\ldots,t$, then we know from the inductive assumption that $Q_{[t-1]}$ is independent of $X_{\cB_{t-1}}$, \Ie  
\begin{equation}
\label{eq:proof-privacy-markov-induction}
	I\left(Q_{[t-1]}; X_{\cB_{t-1}}\right) = 0.
\end{equation}
Then consider
\begin{align*}
	& I\left(Q_{[t]}; X_{\cB_{t}}\right) \\
	& ~~~ = I\left(Q_{[t-1]}; X_{\cB_{t}}\right) + I\left(Q_t; U_t|Q_{[t-1]}\right) \\
	& ~~~~~~ + I\left(Q_t; X_{\cB_{t}\backslash\{\tau,t+1\}}|Q_{[t-1]},U_t\right) \\
	& ~~~ \utag{a}{\leq} I\left(Q_{[t-1]}; X_{\cB_{t-1}}\right) + I\left(Q_t; U_t|Q_{[t-1]}\right) \\
	& ~~~~~~ + I\left(Q_t; X_{\cB_{t}\backslash\{\tau,t+1\}}|Q_{[t-1]},U_t\right) \\
	& ~~~ \utag{b}{=} I\left(Q_t; U_t|Q_{[t-1]}\right) + I\left(Q_t; X_{\cB_{t}\backslash\{\tau,t+1\}}|Q_{[t-1]},U_t\right) \\
	& ~~~ = I\left(Q_t; U_t\right) + I\left(Q_t;Q_{[t-1]}|U_t\right) - I\left(Q_t; Q_{[t-1]}\right) \\
	& ~~~~~~ + I\left(Q_t; X_{\cB_{t}\backslash\{\tau,t+1\}}|Q_{[t-1]},U_t\right) \\
	& ~~~ \leq I\left(Q_t; U_t\right)  + I\left(Q_t; X_{\cB_{t}\backslash\{\tau,t+1\}},Q_{[t-1]}|U_t\right) \\
	& ~~~ \utag{c}{=} I\left(Q_t; U_t\right)  + I\left(X_t; X_{\cB_{t}\backslash\{\tau,t+1\}},Q_{[t-1]}|U_t\right) \\
	& ~~~ = I\left(Q_t; U_t\right)  + I\left(X_t; X_{\cB_{t}\backslash\{\tau,t+1\}}|U_t\right)  \\
	& ~~~~~~ + I\left(X_t; Q_{[t-1]}|X_{\cB_{t}\backslash\{\tau,t+1\}}, U_t\right) \\
	& ~~~ \utag{d}{=} I\left(Q_t; U_t\right)  + I\left(X_t; X_{\cB_{t}\backslash\{\tau,t+1\}}|U_t\right)  \\
	& ~~~ \utag{e}{=} I\left(Q_t; U_t\right)= 0,  
\end{align*}
where \uref{a} follows from $\cB_{t} \subseteq \cB_{t-1}$ by inspecting the definition of $\cB_{t}$,
\uref{b} follows from \eqref{eq:proof-privacy-markov-induction},
\uref{c} follows because $Q_t$ is a stochastic function of $U_t$ and $X_t$,
\uref{d} follows from $X_{\cB_{t-1}}=\left\{X_{\cB_{t}\backslash\{\tau,t+1\}}, U_t, X_t \right\}$ and \eqref{eq:proof-privacy-markov-induction}, 
and \uref{e} follows from the Markovity of $X_t$.

\section{Proof of Proposition~\ref{proposition:lamma-less-one}}
\label{appendix:proposition-bound-one}
From the definitions in \eqref{eq:def-lambda} and \eqref{eq:def-theta}, we can easily see that $\lambda_i$ is non-decreasing with $i$, so $\theta_{i} \geq 0$ for all $i$ if and only if $\lambda_{n-1} \leq 1$. It is sufficient for us to show that $\lambda_{n} \leq 1$.

For any given distribution $\pb{x_t|u_t}$ where $x_t \in \cN$ and $u_t \in \cN^2$, we claim that there exists some $u$ such that 
\begin{equation}
\label{eq:proposition1-proof}
	\pb{x_t=x|u_t=u} \geq \pb{x_t=x|u_t=u_{x,n}}, \forall x \in \cN.
\end{equation}
To see this, one can choose any $u \in \cN^2\backslash\{u_{x,i}: x \in \cN, i=1,\ldots,n-1\}$. Note that since $|\cN^2|=n^2$ and $|\{u_{x,i}: x \in \cN, i=1,\ldots,n-1\}|=n(n-1)$, the set $\cN^2\backslash\{u_{x,i}: x \in \cN, i=1,\ldots,n-1\}$ is non-empty.

By summing \eqref{eq:proposition1-proof} over all $x$
, we have
\begin{equation*}
	\sum_{x \in \cN} \pb{x_t=x|u_t=u} \geq \sum_{x \in \cN} \pb{x_t=x|u_t=u_{x,n}} = \lambda_{n}.
\end{equation*}
Since $\sum_{x \in \cN} \pb{x_t=x|u_t=u} = 1$ for a fixed $u$, we complete showing that $\lambda_{n} \leq 1$. 

\section{Proof of Proposition~\ref{prop:alg-existence}}
\label{Appendix:proof-alg-existence}

Recall that we need to show that for any $u=1,\ldots,m$,
 \begin{equation}
\label{eq:alg-existence-proof}
\begin{aligned}
	 & \sum_{x=1}^{n} \max\left\{\pb{X=x|U=u}- \lambda_{x,n-1}, 0\right\}  \\
	 & ~~~~ \geq \sum_{\ell=1}^{n-1} \sum_{x:u \in \cU_{\ell,x}^{-}} \left( \lambda_{x,\ell} - \lambda_{x,\ell-1}\right).
\end{aligned}
\end{equation}
Assume without loss of generality that $ u = u_{1,\alpha_1} = \cdots = u_{n,\alpha_n}.$
Then, the left-hand side of \eqref{eq:alg-existence-proof}
can be written as
\begin{align*}
	& \sum_{x=1}^{n} \max\left\{\pb{X=x|U=u}-\lambda_{x,n-1}, 0\right\} \nonumber \\
	& = \sum_{x:\alpha_x \geq n-1} \pb{X=x|U=u}-\lambda_{x,n-1} \nonumber \\
	& = \sum_{x:\alpha_x \geq n-1} \left(\lambda_{x,\alpha_x}-\lambda_{x,n-1} \right), 
\end{align*}
and the right-hand side of \eqref{eq:alg-existence-proof} can be written as
\begin{align*}
 \sum_{\ell=1}^{n-1} \sum_{x: u \in \cU_{\ell,x}^{-}} \lambda_{x,\ell} - \lambda_{x,\ell-1} \nonumber& = \sum_{x: \alpha_x \leq n-2} 
    \sum_{\ell = \alpha_{x}+1}^{n-1} \lambda_{x,\ell} - \lambda_{x,\ell-1} \nonumber \\
& = \sum_{x: \alpha_x \leq n-2} 
   \left( \lambda_{x,n-1} - \lambda_{x,\alpha_x} \right). 
\end{align*}
Since
\begin{align}
& \sum_{x:\alpha_x \geq n-1} \left(\lambda_{x,\alpha_x} - \lambda_{x,n-1} \right)- \sum_{x: \alpha_x \leq n-2} \left(\lambda_{x,n-1} - \lambda_{x,\alpha_x} \right) \nonumber \\
& =  \sum_{x=1}^n \left(\lambda_{x,\alpha_x} - \lambda_{x,n-1} \right)  = 1 - \sum_{x=1}^n \lambda_{x,n-1} = \theta_{n}, \label{eq:value-remaining}
\end{align}
we can see that \eqref{eq:alg-existence-proof} is established if and only if $\theta_{n} \geq 0$, which is given by Proposition~\ref{proposition:lamma-less-one}. This completes the proof of Proposition~\ref{prop:alg-existence}.

\begin{remark}
\label{remark:alg-null-matrix}
To benefit the following proof, we give an immediate implication of \eqref{eq:value-remaining} here. As described, the right-hand side of \eqref{eq:alg-existence-proof} is the total values assigned for $\ell=1,\ldots,n-1$
and the left-hand side of \eqref{eq:alg-existence-proof} is the initialization of the matrix $M$, so the remaining values will be assigned for $\ell=n$ as described in \eqref{eq:argument-n} and \eqref{eq:assign-n}. As such, we know from \eqref{eq:value-remaining} that 
\begin{equation}
\label{eq:evaluate-value-n}
\sum_{x: \left(z,x,u \right) \in \cF_{n}} g(z,x,u) = \theta_{n},
\end{equation}
for any $u \in \{1,\ldots,m\}$ and $z=\{1,\ldots,n\}$.
\end{remark}

\section{Proof of Proposition~\ref{proposition:alg-justification}}
\label{Appendix:proposition-justify}
Before proving the proposition, we provide some observations of $\cF_{\ell,x,k}$ for some $x$, $k$ and $\ell=1,\ldots,n-1$ by examining \eqref{eq:argument-seperate}. Let $\mathbbm{1}\{\cdot\}$ be the indicator function.
	\begin{itemize}
		\item If $\left(\bar{z},\bar{x},\bar{u}\right) \in \cF_{\ell,x,k}$ for some $\ell$, $x$ and $k$, then $\bar{z} = z_{\ell,x,k}$ is uniquely determined, \Ie if $\left(\bar{z},\bar{x},\bar{u}\right) \in \cF_{\ell,x,k}$, then 
		\begin{equation}
		\label{eq:unique-z}
		\left(z',x',u'\right) \notin \cF_{\ell,x,k}, \forall z' \neq \bar{z}.
		\end{equation}

		\item For any $\left(\bar{z},\bar{x},\bar{u}\right) \in \cF_{\ell,x,k}$,  $|\bar{z}| = \ell$,
		and $\bar{x} \in \bar{z}$.

		\item
		 The cardinality of each $\cF_{\ell,x,k}$ is $|\cF_{\ell,x,k}| = m$. In particular, all tuples $\left(\bar{z},\bar{x},\bar{u}\right) \in \cF_{\ell,x,k}$ have distinct values of $\bar{u}$. In other words, let $\left(\bar{z}, \cdot,\cdot\right) \in \cF_{\ell,x,k}$ denote that there exists some $\bar{x},\bar{u}$ such that $\left(\bar{z}, \bar{x},\bar{u} \right) \in \cF_{\ell,x,k}$, and then 
		\begin{equation}
			\label{eq:alg-indicator-drop-u}
		\sum_{x'=1}^n \mathbbm{1}\left\{\left(\bar{z},x',u'\right) \in \cF_{\ell,x,k}\right\} 
		= \mathbbm{1}\left\{\left(\bar{z}, \cdot,\cdot\right) \in \cF_{\ell,x,k} \right\},
		\end{equation}
		for any $u' \in \{1,\ldots,m\}$.
	\end{itemize}

\begin{enumerate}
\item The first statement is straightforward. Suppose that $\left(\bar{z},\bar{x},\bar{u}\right) \in \cF_{\ell}$ for some $\ell$. If $\ell=1,\ldots,n-1$, we know from \eqref{eq:argument-l} that $\cF_{\ell}$ is the union of $\cF_{\ell,x,k}$. For each $\cF_{\ell,x,k}$, we know that $\bar{x} \in \bar{z}$ for any $\left(\bar{z},\bar{x},\bar{u}\right) \in \cF_{\ell,x,k}$. If $\ell=n$,  $\bar{z} = \cN$ for all $\left(\bar{z},\bar{x},\bar{u}\right) \in \cF_{n}$, so $\bar{x} \in \bar{z}$.

\item For the second statement, when $|\bar{z}|= \ell \in \{1,\ldots,n-1\}$, 
\begin{align}
	 &\sum_{x:\left(\bar{z},x,\bar{u}\right) \in \cF} g\left(\bar{z},x,\bar{u}\right) 
	 = \sum_{x:\left(\bar{z},x,\bar{u}\right) \in \cF_{\ell}}  g\left(\bar{z},x,\bar{u}\right) \nonumber \\
	& \utag{a}{=} \sum_{x} \sum_{x'=1}^{n} \sum_{k:\left(\bar{z},x,\bar{u}\right) \in \cF_{\ell,x',k}}  \nu_{\ell,x',k} \nonumber \\
	& \utag{b}{=} \sum_{x'=1}^{n} \sum_{k=1}^{c_{\ell,x'}}  \nu_{\ell,x',k} \cdot \mathbbm{1}\left\{\left(\bar{z},\cdot,\cdot\right) \in \cF_{\ell,x',k}  \right\}, \label{eq:alg-privacy-algorithm-justification}
\end{align}
where \uref{a} follows from \eqref{eq:assign-l} and \uref{b} follows from \eqref{eq:alg-indicator-drop-u}. 

When $|\bar{z}|=\ell=n$, \Ie $\bar{z}=\{1,\ldots,n\}$, we have 
\begin{align}
 \sum_{x:\left(\bar{z},x,\bar{u}\right) \in \cF} g\left(\bar{z},x,\bar{u}\right) 
	& = \sum_{x:\left(\bar{z},x,\bar{u}\right) \in \cF_{n}}  g\left(\bar{z},x,\bar{u}\right) 
	 \utag{c}{=} \theta_{n}, \label{eq:proposition-state2-n}
\end{align}
where \uref{c} follows from \eqref{eq:evaluate-value-n}.

Since both the right-hand sides of  \eqref{eq:alg-privacy-algorithm-justification} and \eqref{eq:proposition-state2-n} are independent of $\bar{u}$, for any given $\bar{z}$, $\bar{u}$ and $\bar{u}'$
\begin{equation*}
	\sum_{x:\left(\bar{z},x,\bar{u}\right) \in \cF} g\left(\bar{z},x,\bar{u}\right) = \sum_{x:\left(\bar{z},x,\bar{u}'\right) \in \cF} g\left(\bar{z},x,\bar{u}'\right).
\end{equation*}

\item

As for the third statement, for any given $\bar{x}$ and $\bar{u}$, 
\begin{align}
	& \sum_{z:\left(z,\bar{x},\bar{u}\right) \in \cF} g\left(z,\bar{x},\bar{u}\right)  = \sum_{\ell=1}^{n} \sum_{z:\left(z,\bar{x},\bar{u}\right) \in \cF_{\ell}} g\left(z,\bar{x},\bar{u}\right) \nonumber \\
	& = \sum_{\ell=1}^{n-1} \sum_{z:\left(z,\bar{x},\bar{u}\right) \in \cF_{\ell}} g\left(z,\bar{x},\bar{u}\right) +  \sum_{z:\left(z,\bar{x},\bar{u}\right) \in \cF_{n}} g\left(z,\bar{x},\bar{u}\right). \label{eq:statement3-two-terms}
\end{align}

For the first term of \eqref{eq:statement3-two-terms}, we have
\begin{align}
	&  \sum_{\ell=1}^{n-1} \sum_{z:\left(z,\bar{x},\bar{u}\right) \in \cF_{\ell}} g\left(z,\bar{x},\bar{u}\right)   \nonumber \\
	& \utag{a}{=} \sum_{\ell=1}^{n-1} \sum_{z} \sum_{x=1}^{n} \sum_{k:\left(z,\bar{x},\bar{u}\right) \in \cF_{\ell,x,k}}  \nu_{\ell,x,k}  \nonumber \\
	& \utag{b}{=} \sum_{\ell=1}^{n-1} \sum_{x=1}^{n}  \sum_{k:\left(\cdot,\bar{x},\bar{u}\right) \in \cF_{\ell,x,k}}  \nu_{\ell,x,k}  , \label{eq:alg-prob-dis} 
\end{align}
where \uref{a} follows by substituting \eqref{eq:assign-l}, and \uref{b} follows from \eqref{eq:unique-z}. 

By examining $\cF_{\ell,x,k}$, we can see two disjoint subsets,
\begin{multline*}
	  \cF_{\ell,x,k}^{-} 
	  := \{\left(\bar{z},\bar{x},\bar{u}\right): \bar{z}=z_{\ell,x,k}, \bar{x}=\zeta_{\ell,x,k}(i), \\\bar{u}=u_{x,i} \in \cU_{\ell,x}^{-}  \}, 	
\end{multline*}
and
\begin{equation*}
	  \cF_{\ell,x,k}^{+} :=  \left\{\left(\bar{z},\bar{x},\bar{u}\right):\bar{z}=z_{\ell,x,k},\bar{x}=x,\bar{u} \in \cU_{\ell,x}^{+} \right\},
\end{equation*} 

For a fixed $\bar{u}$, assume that 
$
    \bar{u} = u_{1,\alpha_1} = \cdots = u_{n,\alpha_n}.
$
Then, we write \eqref{eq:alg-prob-dis} as 
\begin{align}
	& \sum_{\ell=1}^{n-1} \sum_{z:\left(z,\bar{x},\bar{u}\right) \in \cF_{\ell}} g\left(z,\bar{x},\bar{u}\right)  \nonumber  = \sum_{\ell=1}^{n-1} \sum_{x=1}^{n}  \sum_{k:\left(\cdot,\bar{x},\bar{u}\right) \in \cF_{\ell,x,k}} \hspace{-5pt} \nu_{\ell,x,k} \nonumber \\
	& = \sum_{\ell=1}^{n-1} \sum_{x:\alpha_x \leq \ell-1} \sum_{k:\bar{x} = \zeta_{\ell,x,k}(\alpha_x) }  \nu_{\ell,x,k} + \sum_{\ell=1}^{n-1}  \sum_{k: \alpha_{\bar{x}} \geq \ell}  \nu_{\ell,\bar{x},k}  \nonumber \\
	& = \sum_{\ell=1}^{n-1} \sum_{x:\alpha_x \leq \ell-1} \sum_{k:\bar{x} = \zeta_{\ell,x,k}(\alpha_x) } \hspace{-10pt} \nu_{\ell,x,k} +\hspace{-10pt} \sum_{\ell=1}^{\min\{n-1,\alpha_{\bar{x}}\}}  \hspace{-5pt}\sum_{k} \nu_{\ell,\bar{x},k}.  \label{eq:alg-prob-dis-two-terms}
\end{align}
From \eqref{eq:alg-sum-nu}, we know that 
$
\sum_{k} \nu_{\ell,\bar{x},k} =  \lambda_{\bar{x},\ell} - \lambda_{\bar{x},\ell-1} ,
$
and hence the second term of \eqref{eq:alg-prob-dis-two-terms} can be written as
\begin{align}
	\sum_{\ell=1}^{\min\{n-1,\alpha_{\bar{x}}\}}  \sum_{k} \nu_{\ell,\bar{x},k} 
	& = \lambda_{\bar{x},\min\{n-1,\alpha_{\bar{x}}\}}. \label{eq:alg-prob-dis-term2}
\end{align}

For the first term of \eqref{eq:alg-prob-dis-two-terms}, we know from \eqref{eq:alg-multiset-value-zx} that
	\begin{equation*}
		\sum_{k: \zeta_{\ell,x,k}(\alpha_x)=\bar{x}} \nu_{\ell,x,k} = v_{\ell,x,\alpha_x,j}, 
	\end{equation*}
and $\bar{x} = e_{\ell,x,i,j}$ for some $j$, where $v_{\ell,x,\alpha_x,j}$ and  $e_{\ell,x,i,j}$ are defined in \eqref{eq:alg-row-compensate}. Then, we know from \eqref{eq:alg-update} that  
	\begin{equation*}
		 \sum_{k: \zeta_{\ell,x,k}(\alpha_x)=\bar{x}} \nu_{\ell,x,k} = v_{\ell,x,\alpha_x,j} = M_{u_{x,\alpha_x},\bar{x}}^{-} = M_{\bar{u},\bar{x}}^{-},
	\end{equation*}
\Ie the value subtracted from $M_{\bar{u},\bar{x}}$ for given $\ell$ and $x$. Thus, we have
\begin{align}
	& \sum_{\ell=1}^{n-1} \sum_{x:\alpha_x \leq \ell-1} \sum_{k:\bar{x} = \zeta_{\ell,x,k}(\alpha_x) }  \nu_{\ell,x,k}   = \sum_{\ell=1}^{n-1} \sum_{x:\alpha_x \leq \ell-1} M_{\bar{u},\bar{x}}^{-} \nonumber\\
	& = \sum_{\ell=1}^{n-1} \sum_{x: \bar{u} \in \cU_{\ell,x}^{-}} M_{\bar{u},\bar{x}}^{-},
	\label{eq:alg-prob-dis-term1}
\end{align}
\Ie all values subtracted from $M_{\bar{u},\bar{x}}$ for $\ell=1,\ldots,n-1$ and all $x$.	
By substituting \eqref{eq:alg-prob-dis-term1} and \eqref{eq:alg-prob-dis-term2} in \eqref{eq:alg-prob-dis-two-terms}, we have 
\begin{multline}
 \sum_{\ell=1}^{n-1} \sum_{z:\left(z,\bar{x},\bar{u}\right) \in \cF_{\ell}} g\left(z,\bar{x},\bar{u}\right)   \\
=  \sum_{\ell=1}^{n-1} \sum_{x: \bar{u} \in \cU_{\ell,x}^{-}} M_{\bar{u},\bar{x}}^{-} + \lambda_{\bar{x},\min\{n-1,\alpha_{\bar{x}}\}}. \label{eq:statement3-first-term}
\end{multline}

Then, substituting \eqref{eq:statement3-first-term} in \eqref{eq:statement3-two-terms}, we have 
\begin{align*}
	& \sum_{z:\left(z,\bar{x},\bar{u}\right) \in \cF} g\left(z,\bar{x},\bar{u}\right)  \\
	& = \sum_{\ell=1}^{n-1} \sum_{z:\left(z,\bar{x},\bar{u}\right) \in \cF_{\ell}} g\left(z,\bar{x},\bar{u}\right) +  \sum_{z:\left(z,\bar{x},\bar{u}\right) \in \cF_{n}} g\left(z,\bar{x},\bar{u}\right)  \nonumber \\
	& =
	\sum_{\ell=1}^{n-1} \sum_{x: \bar{u} \in \cU_{\ell,x}^{-}} M_{\bar{u},\bar{x}}^{-} + \lambda_{\bar{x},\min\{n-1,\alpha_{\bar{x}}\}} +\\[-10pt]
	&\hspace{130pt}\sum_{z:\left(z,\bar{x},\bar{u}\right) \in \cF_{n}} g\left(z,\bar{x},\bar{u}\right).  
\end{align*}

Recalling that we assign all the remaining values $M_{\bar{u},\bar{x}}$ to $g\left(z,\bar{x},\bar{u}\right)$ in \eqref{eq:argument-n} and \eqref{eq:assign-n} when $\ell=n$, we know that 
\begin{equation*}
	\begin{aligned}
		& \sum_{\ell=1}^{n-1} \sum_{x: \bar{u} \in \cU_{\ell,x}^{-}} M_{\bar{u},\bar{x}}^{-} + \sum_{z:\left(z,\bar{x},\bar{u}\right) \in \cF_{n}} g\left(z,\bar{x},\bar{u}\right) \\
		& = \max\left\{\pb{X=\bar{x}|U=\bar{u}}- \lambda_{\bar{x},n-1}, 0\right\},
	\end{aligned}
\end{equation*}
\Ie the initial value of $M_{\bar{u},\bar{x}}$ defined in \eqref{eq:alg-initialize}. 
Therefore,  
\begin{align*}
	& \sum_{z:\left(z,\bar{x},\bar{u}\right) \in \cF} g\left(z,\bar{x},\bar{u}\right)  \\
	& = \sum_{\ell=1}^{n-1} \sum_{x: \bar{u} \in \cU_{\ell,x}^{-}} M_{\bar{u},\bar{x}}^{-} + \lambda_{\bar{x},\min\{n-1,\alpha_{\bar{x}}\}} + \\[-10pt] &\hspace{130pt}\sum_{z:\left(z,\bar{x},\bar{u}\right) \in \cF_{n}}g\left(z,\bar{x},\bar{u}\right) \\
	& = \max\left\{\pb{X=\bar{x}|U=\bar{u}}- \lambda_{\bar{x},n-1}, 0\right\} +\\ &\hspace{160pt}\lambda_{\bar{x},\min\{n-1,\alpha_{\bar{x}}\}} \\
	& = \max\left\{\pb{\bar{x}|\bar{u}}- \lambda_{\bar{x},n-1}, 0\right\} 
	+ \min\left\{ \lambda_{\bar{x},n-1},    \pb{\bar{x}|\bar{u}}\right\} \\
	& = \pb{\bar{x}|\bar{u}},
\end{align*}
which completes the proof.

\end{enumerate}

\section{Proof of Proposition~\ref{proposition:alg-performance}}
\label{Appendix:proposition-perform}
The proof is quite straightforward from previous intermediate steps. As we know that $\pb{z|u} = \pb{z|u'}$ for any $u, u' \in \cN^2$ and $z \in \cZ$ from Proposition~\ref{proposition:alg-justification}, for any given $\ell \in \{1,\ldots,n-1\}$, we have
\begin{align*}
	\sum_{z:|z|=\ell} \pb{z} & = \sum_{z:|z|=\ell} \pb{z|u}  = \sum_{z:|z|=\ell} \sum_{x:\left(z,x,u \right) \in \cF} g\left(z,x,u\right) \\
	& \utag{a}{=} \sum_{z:|z|=\ell} \sum_{x'=1}^{n} \sum_{k=1}^{c_{\ell,x'}}  \nu_{\ell,x',k} \cdot \mathbbm{1}\left\{\left(z,\cdot,\cdot\right) \in \cF_{\ell,x',k}  \right\} \\
	& = \sum_{x'=1}^{n} \sum_{k=1}^{c_{\ell,x'}}  \nu_{\ell,x',k} \cdot \sum_{z:|z|=\ell} \mathbbm{1}\left\{\left(z,\cdot,\cdot\right) \in \cF_{\ell,x',k}  \right\} \\
	& \utag{b}{=} \sum_{x'=1}^{n} \sum_{k=1}^{c_{\ell,x'}}  \nu_{\ell,x',k}  \utag{c}{=} \sum_{x'=1}^{n} \lambda_{x',\ell} - \lambda_{x',\ell-1} =\theta_{\ell},
\end{align*}
where \uref{a} follows from	\eqref{eq:alg-privacy-algorithm-justification}, \uref{b} follows from \eqref{eq:unique-z}, and \uref{c} follows from \eqref{eq:alg-sum-nu}.
For $\ell=n$, we have 
\begin{equation*}
	\sum_{z:|z|=n} \pb{z}= 1- \sum_{\ell=1}^{n-1} \sum_{z:|z|=\ell} \pb{z} = 1- \sum_{\ell=1}^{n-1}\theta_{\ell} = \theta_{n}
\end{equation*}
by definition, then
$\pb{|Z|=\ell} = \theta_{\ell}$, for all $\ell~=~1,\dots,n$.

\section{Proof of Lemma~\ref{lemma:outer}}
\label{Appendix:lemma-outer}
Consider
    \begin{align*}
        \max_{u \in \cN^2} p\left(x|u\right) 
        & \utag{a}{=} \max_{u \in \cN} \sum_{y:x \in y} p\left(x, y|u\right) \\
        & = \max_{u \in \cN^2} \sum_{y:x \in y} p\left(y|u\right) p\left(x|y,u\right) \\
        & \utag{b}{=} \max_{u \in \cN^2} \sum_{y:x \in y} p\left(y\right) p\left(x| y,u\right) \\
        & \leq \sum_{y:x \in y} p\left(y\right) \max_{u \in \cN^2}  p\left(x| y,u\right) \\
        & \leq \sum_{y:x \in y} p\left(y\right), 
    \end{align*}
where \uref{a} follows from $p(x,y|u)=0$ for $x \notin y$, and \uref{b} follows because $Y$ is independent of $U$.
Thus, we obtain that 
\begin{align*}
     \sum_{x \in \cN} \max_{u \in \cN^2} p\left(x|u\right) 
    & \leq  \sum_{x \in \cN} \sum_{y:x \in y} p\left(y\right) 
     =  \sum_{y \in \sP\left(\cN\right)} \sum_{x:x \in y} p\left(y\right) \\
    & =  \sum_{y \in \sP\left(\cN\right)} p\left(y\right) \sum_{x:x \in y} 1 
     = \mathbb{E}\left[|Y|\right],
\end{align*}
which completes the proof.

\end{document}